\documentclass[11pt,letterpaper]{article}
\usepackage[utf8]{inputenc}
\usepackage[margin=1in]{geometry}
\usepackage{array}
\usepackage{multirow}
\usepackage{amsmath,amssymb,amsthm,amsfonts,bm}
\usepackage[lined,ruled,,boxed, vlined, linesnumbered, algo2e]{algorithm2e}

\usepackage{nicefrac}
\usepackage{mathtools}
\usepackage[colorlinks,citecolor=blue,bookmarks=true]{hyperref}

\usepackage{tikz}

\usepackage{caption}
\usepackage{subcaption}
\usepackage{tikzsymbols}
\usepackage{tabularx} 
\usepackage{dsfont}
\usepackage[inline]{enumitem}
\theoremstyle{plain}
\newtheorem{theorem}{Theorem}[section]
\newtheorem{lemma}[theorem]{Lemma}
\newtheorem{corollary}[theorem]{Corollary}

\theoremstyle{definition}

\newcommand{\ignore}[1]{}

\newcommand{\CS}[1]{{\mathrm{#1}}}

\newcommand{\E}[1]{\text{\normalfont E}\left[ #1 \right]}
\makeatletter
\newcommand{\@giventhatnostar}[2]{#1\;\middle|\;#2}
\newcommand{\@giventhatstar}[3][]{#1(#2\;#1|\;#3#1)}
\newcommand{\given}{\@ifstar\@giventhatstar\@giventhatnostar}

\makeatother

\global\long\def\IID{\text{i.i.d.}}
\global\long\def\ALG{\text{\normalfont ALG}}
\global\long\def\OPT{\text{\normalfont OPT}}

\DeclareMathOperator*{\argmin}{arg\,min}

\newcommand{\rofl}[1][q]{{#1}\textsc{DistProb}}

\newcommand{\pul}[2]{{p(#2,#1)}}
\newcommand{\pful}[3]{{p_{#1}(#3,#2)}}
\newcommand{\coin}[1]{\mathsf{p}(#1)}
\newcommand{\distalg}[1]{\mathsf{d}(#1)}
\newcommand{\distletter}{x}
\newcommand{\demands}{{\CS{U}}}
\newcommand{\optcluster}{{\CS{C}^*}}
\newcommand{\optcenter}{{c^*}}
\newcommand{\cpoints}[1]{{C^*_{\geq #1}}}
\newcommand{\tcpoints}[1]{{C^*_{ \left[#1, T_{#1}\right]}}}
\newcommand{\uptocpoints}[1]{{C^*_{\leq #1}}}

\newcommand{\closer}[2]{ {\cpoints{#1}\left( #2 \right)} }
\newcommand{\event}{ {\mathcal{E}} }

\newcommand{\balanced}[1]{ {\text{Bal}\left(#1\right)} }
\newcommand{\imbalanced}[1]{ {\text{ImBal}\left( #1 \right)} }
\newcommand{\nofacility}[1]{ {\text{N}\left( #1 \right)} }
\usepackage{tikz}
\usepackage{multicol}

\newcommand{\cI}{{\cal I}}

\makeatletter
\newenvironment{proofw}{\par
  \pushQED{\qed}%
  \normalfont \topsep7\p@\@plus6\p@\relax
  \trivlist
  \item[]\ignorespaces
}{%
  \popQED\endtrivlist\@endpefalse
}
\makeatother

\title{Almost Tight Bounds for Online Facility Location in the Random-Order Model}

\author{
    Haim Kaplan%
    \thanks{Blavatnik School of Computer Science, Tel Aviv University, Israel.
    Email: \texttt{haimk@tau.ac.il}. 
    Supported by ISF grant.\ 1595-19 and the Blavatnik Family Foundation.}
    \and
    David Naori%
    \thanks{Computer Science Department, Technion, Israel.
    Emails: \texttt{\{dnaori,danny\}@cs.technion.ac.il}.} 
    \and Danny Raz%
    \footnotemark[2]
}
\date{}

\begin{document}
\maketitle
\begin{abstract}
We study the online facility location problem with uniform facility costs in the random-order model. Meyerson's algorithm [FOCS'01] is arguably the most natural and simple online algorithm for the problem with several advantages and appealing properties. Its analysis in the random-order model is one of the cornerstones of random-order analysis beyond the secretary problem. Meyerson's algorithm was shown to be  (asymptotically) optimal in the standard worst-case adversarial-order model and $8$-competitive in the random order model. While this bound in the random-order model is the long-standing state-of-the-art, it is not known to be tight, and the true competitive-ratio of Meyerson's algorithm remained an open question for more than two decades. 

We resolve this question and prove tight bounds on the competitive-ratio of Meyerson's algorithm in the random-order model, showing that it is exactly $4$-competitive. Following our tight analysis, we introduce a generic parameterized version of Meyerson's algorithm that retains all the advantages of the original version. We show that the best algorithm in this family is exactly $3$-competitive. On the other hand, we show that no online algorithm for this problem can achieve a competitive-ratio better than $2$. Finally, we prove that the algorithms in this family are robust to partial adversarial arrival orders.
\end{abstract}

\section{Introduction}\label{sec:intro}

In the classical online metric uncapacitated facility location problem, we have a metric space where facilities can be opened at any point for a given cost (uniform facility cost). A sequence of demand points arrive one by one over time, and upon arrival of a demand point, it must be irrevocably assigned to an open facility. A demand point can either be assigned to an existing open facility, or a new facility can be opened for this purpose. The cost of assigning a demand point to a facility is the distance between the demand point and the facility. The goal is to minimize the total cost (assignment cost and facility opening cost) paid for serving all the demand points.

This problem was first considered by Meyerson~\cite{meyerson2001online}, and it has been studied extensively since (see e.g.,~\cite{anagnostopoulos2004simple, fotakis2003competitive, fotakis2007primal, jiang2021online, almanza2021online, fotakis2021learning, azar2022online} and the survey by Fotakis~\cite{fotakis2011online} and references therein). In his seminal paper~\cite{meyerson2001online}, Meyerson studied the problem both in the standard worst-case (adversarial order) model, and mainly in the random-order model, which is particularly suitable for many applications of the facility location problem.

Meyerson considered what is arguably the most simple and natural online algorithm for the facility location problem, which we call $\rofl[]$: When a demand point arrives, $\rofl[]$ randomly decides whether to assign it to an existing open facility, or to open a new facility at the demand point. The decision is based only on the distance, $d$, between the demand point and the nearest open facility (which is the assignment cost to this facility), and the facility opening cost $f$. $\rofl[]$ opens a new facility with probability $\min\{d/f,1\}$, and otherwise, it assigns the demand point to the nearest open facility.

$\rofl[]$ has many advantages: it is simple and intuitive, very efficient computationally and memoryless. It has inspired many studies, and it is used as a building block in algorithms for other online problems (see e.g.,~\cite{guha2003clustering, charikar2003better,  fotakis2007primal, fotakis2011online, fotakis2021learning}). Meyerson's elegant analysis of $\rofl[]$ is one of the cornerstones of random-order analysis, and is sometimes taught in algorithms classes~\cite{lecture_roughgarden, gupta2022random}.

In the worst-case adversarial-order model, the facility location problem is considered to be solved. Shortly after the seminal paper by Meyerson, Fotakis showed that $\rofl[]$ is optimal $O(\log{n}/\log{\log{n}})$-competitive in the worst-case~\cite{fotakis2003competitive}. In contrast, this is not the case in the random-order model. In~\cite{meyerson2001online}, Meyerson showed that $\rofl[]$ is $8$-competitive, and it remained the best-known competitive-ratio for the problem until now.

In some cases, tailoring an optimal algorithm in the random-order model compromises its performance in the worst-case, and vice versa~\cite{kaplan2020competitive}. Analyzing a single algorithm in both the worst-case and the random-order model provides a deeper understanding of its performance in different conditions, and sets new standards in the design of online algorithms~\cite{mirrokni2012simultaneous, raghvendra2016robust, molinaro2017online, gupta2021random}. The ideal is to have the ``best-of-both-worlds'' algorithm, i.e., a single algorithm with good performance guarantees in both models. $\rofl[]$ is a good example for this, it provides the best performance guarantee in the worst-case (asymptotically), and when the online sequence arrives in random order, it provides a much better guarantee. 

In this paper, we study the online facility location problem in the random order model. We provide tight bounds on the competitive-ratio of $\rofl[]$. We then design an improved algorithm that admits a better competitive-ratio while maintaining all the advantages of $\rofl[]$. We show that the performance of our improved algorithm is close to optimal by proving a lower bound on the performance of any algorithm for the problem. We also prove that the algorithms that we consider have a nice robustness property to partial adversarial orders. High-level details follow.

\subsection{Our Contribution}

We provide tight bounds on the performance of $\rofl[]$ in the random-order model, and show that its exact competitive-ratio is $4$. 

Our analysis maintains the elegance in Meyerson's analysis, and even simplifies some aspects of it. Like Meyerson's analysis, our analysis is done per cluster of the optimal solution. The crux of our analysis lies in pinpointing a probabilistic event of opening a facility which is, roughly speaking, well placed among the remaining demand points in the cluster. This opened facility provides an upper bound on the expected distances between the following demand points and their closest open facility (which, in turn, upper bounds the expected cost incurred by serving these demand points). On the other hand, we also need to upper bound the cost incurred by serving the demand points that arrive before this event occurs. The choice of this event carefully balances these two things.

Our analysis sheds new light on the trade-off between assigning demand points to existing  facilities at a low cost, and opening new facilities for future use. It allows us to generalize $\rofl[]$ and consider a family of algorithms that open facilities with different probabilities. 

More concretely, we consider a generalized version of $\rofl[]$ that when a demand point arrives, instead of opening a facility with probability $\min\{d / f, 1\}$, it opens a facility with probability $g(d)$, for some function $g : \mathbb{R}_{\geq 0} \rightarrow [0,1]$. We show that the best functions has the form $g(d) = \min\{ q \cdot d / f , 1\}$ for some $q \in (0,1)$. Our analysis generalizes to provide tight bounds on the competitive-ratio of this generic algorithm for any value of $q$. The best competitive-ratio of an algorithm in this family is $3$ and it is obtained for $q=1/2$. We call this algorithm $\rofl[\nicefrac{1}{2}]$. We note that $\rofl[\nicefrac{1}{2}]$ is still (asymptotically) optimal in the worst-case adversarial-order model, and therefore it has the ``best-of-both-worlds'' property.

We then prove a lower-bound of $2$ on the competitive-ratio of any algorithm for the facility location problem, which applies even in the weaker online $\IID$ model. In the $\IID$ model, the demand points are drawn independently from a probability distribution over the points in the metric space, and the online algorithm has full prior knowledge of this distribution. 

Finally, we study how well $\rofl[]$ performs in a partial adversarial order setting. We consider a setting in which 
for a parameter $\rho \in (0,1)$,
an adversarially chosen $(1-\rho)$-fraction of the demand points arrive in adversarial order. The remaining demand points in each cluster of the optimal solution are injected in random positions between the adversarially ordered points of the cluster. In this setting, we show that the competitive-ratio of $\rofl[]$, is within a factor of at most $(2-\rho)/\rho$ from its random-order competitive-ratio. For instance, our analysis shows that $\rofl[\nicefrac{1}{2}]$ is $3.66$-competitive if $10\%$ of the demand points arrive in adversarial order. This result can be seen as part of a growing body of work on robust random-order algorithms and non-uniform arrival orders that aim to weaken the random order assumption~\cite{kesselheim2015secretary, bradac2020robust, argue2022robust, kesselheim2020knapsack}.
\subsection{Additional Related Work}
In recent years, the random-order model has been widely adopted for the design and analysis of online algorithms for various online problems. Some recent examples are the set cover problem~\cite{gupta2022random}, edge coloring~\cite{bhattacharya2021online, bahmani2012online}, weighted bipartite matching, and various other packing problems that generalize the classical secretary problem~\cite{DBLP:conf/esa/KesselheimRTV13, kesselheim2018primal, kaplan2022online, korula2009algorithms, FeldmanSZ18, naori2019online, albers2021improved}. See also the survey by Gupta and Singla~\cite{gupta2021random} and references therein. 

Following the seminal paper by Meyerson~\cite{meyerson2001online}, the online facility location problem was mostly considered in the standard worst-case adversarial-order model. In~\cite{fotakis2003competitive}, Fotakis gave a lower bound of $\Omega(\log{n}/\log{\log{n}})$, and noted that $\rofl[]$ achieves this bound. Fotakis also presented a deterministic $O(\log{n}/\log\log{n})$-competitive algorithm. Anagnostopoulos et al.~\cite{anagnostopoulos2004simple} presented a simpler and more computationally efficient deterministic online algorithm, that achieves a worst-case competitive-ratio of $O(\log{n})$ for Euclidean spaces of constant dimension. On the other hand, their algorithm is not constant competitive in the random-order model.

In~\cite{fotakis2007primal}, Fotakis presented a simple deterministic $O(\log{n})$-competitive algorithm for the online facility location problem which is guided by the dual of an LP relaxation for the problem. Later, Nagarajan and Williamson~\cite{nagarajan2013offline} presented an elegant dual-fitting analysis of Fotakis' algorithm which also proves a competitive-ratio of $O(\log{n})$. They adapted Fotakis' algorithm to the more general online facility leasing problem. For other variants of the facility location problem see the survey by Fotakis~\cite{fotakis2011online}. More recently, Cygan et al.~\cite{cygan2018online} modified $\rofl[]$ to a setting in which demand points may depart.

In a work related to our results on $\rofl[]$ in a partial adversarial order setting, Lang \cite{lang2018online} studied $\rofl[]$ in the $t$-semi-random order setting. In this setting, the demand points are initially ordered uniformly at random. Then, the random-order sequence can be manipulated by a $t$-bounded adversary. This means that at each point in time, the adversary holds a set of $t$ demand points from which it can choose the next demand point to arrive in the online sequence. Initially, the adversary gets the first $t$ demand points in the random-order sequence from which it selects the first demand point to arrive in the online sequence. Then, at each online round, the next demand point from the random-order sequence is added to the set of demand points from which the adversary selects the next demand point to arrive in the online sequence.

Lang shows that $\rofl[]$ is $O(\log{t}/ \log{\log{t}})$-competitive in the $t$-semi-random order setting, and gives a matching lower bound on the competitive-ratio of any algorithm for the facility location problem in the $t$-semi-random setting. We note that our partial adversarial order setting and the $t$-semi-random setting are not directly comparable. For example, in our partial adversarial order setting with $\rho = 1/2$, the adversary can make sure that half of the demand points in the input (which the adversary can choose) will always arrive in the same relative order in the online sequence. To achieve this in the $t$-semi-random setting, we need $t = n/2$,  and no online algorithm can achieve a constant competitive-ratio for this choice of $t$ (while $\rofl[]$ achieves a constant competitive-ratio in our partial adversarial order setting with $\rho = 1/2$). On the other hand, in the $t$-semi-random setting with $t=2$, the adversary can make sure that two demand points $u_1,u_2$ will always arrive consecutively in the online sequence. This cannot be achieved in the partial adversarial order setting that we study in this paper.

The facility location problem has received much recent attention in online settings with predictions~\cite{fotakis2021learning, jiang2021online,azar2022online,almanza2021online}. For instance, in~\cite{fotakis2021learning, jiang2021online}, upon the arrival of a demand point, the algorithm receives a prediction on the facility that the demand point is assigned to in the optimal solution. 
The goal is to have an algorithm that uses the predictions, and obtains a performance guarantee that depends on the prediction error: It should be better than the best online worst-case performance guarantee when the predictions are accurate, and close to it when the predictions are erroneous. The works in~\cite{fotakis2021learning, jiang2021online,azar2022online,almanza2021online} achieve this goal with different prediction types and different performance guarantees. We note that the algorithms in~\cite{jiang2021online, fotakis2021learning, almanza2021online} are based on Meyerson's algorithm. 
\subsection{Organization of the paper}
In Section~\ref{sec:def} we give a formal definition of the online facility location problem in the random-order model and establish notations.
In Section~\ref{sec:meyerson} we present our tight analysis of $\rofl[]$ in the random-order model.
In Section~\ref{sec:improve}, we introduce our generic version of $\rofl[]$, and present tight bounds on its performance.
In Section~\ref{sec:lowerbound} we prove the hardness result for any algorithm for online facility location in the $\IID$ model. In Section~\ref{sec:mixed} we introduce the setting of partial adversarial arrival order, and prove a robustness result for our considered algorithms in this setting. Finally, we conclude and discuss open questions in Section~\ref{sec:discussion}.
\section{Problem Definition}\label{sec:def}
In the (metric, uncapacitated) facility location problem, we are given a metric space $(\CS{M},d)$ where $\CS{M}$ is the set of points, and $d : \CS{M} \times \CS{M} \rightarrow \mathbb{R}_{\geq 0}$ is a non-negative and symmetric distance function that satisfies the triangle inequality. We are also given a multiset of demand points $\demands$ in $\CS{M}$, and a facility opening cost $f \in \mathbb{R}_{\geq 0}$ (uniform facility cost). 

Each demand must be assigned to an open facility. A facility can be opened at any point in the metric space for a cost of $f$. For $\CS{F} \subseteq \CS{M}$ and $v \in \CS{M}$, $d(\CS{F},v)$ is the minimal distance between a point in $\CS{F}$ and $v$, that is, $d(\CS{F},v) = \min_{u \in \CS{F}} d(u,v)$. We define $d(\emptyset,v) = \infty$. Given a set of open facilities $\CS{F} \subseteq \CS{M}$, the cost of assigning a demand $v$ is $d(\CS{F},v)$. The goal is to find a set of facilities $\CS{F} \subseteq \CS{M}$ that minimizes the total cost $|\CS{F}| \cdot f + \sum_{v\in \demands} d(\CS{F},v)$. By scaling, we assume throughout this paper, without loss of generality, that $f=1$.

In the online version of the problem, the demands in $\demands$ arrive one by one. Let $v_1,\dots,v_n$ denote the input sequence. The arrival order is determined by the online model, which we specify thereafter. When a demand point $v_\ell$ arrives, the online algorithm must decide immediately, whether and where to open new facilities. Then, $v_\ell$ is irrevocably assigned to its nearest open facility. Let $F_\ell$ be the set of all facilities that the algorithm opens by the end of online round $\ell \in [n]$ (note that $F_\ell$ may be a random variable).  The algorithm's \textit{service cost} for $v_\ell$ is the cost incurred by the algorithm at online round $\ell$, i.e., the facility opening cost and the assignment cost, $(|F_{\ell}| - |F_{\ell-1}|) +  d(F_\ell,v_\ell)$. The total cost incurred by the algorithm is given by $|F_n| + \sum_{\ell=1}^n d(F_\ell,v_\ell)$.

In the online random-order model, the input sequence $v_1,\dots,v_n$ is a uniformly random permutation of the demand points in $\demands$.\footnote{We note that in contrast to other problems in the random-order model, we do not need to assume that $n$ is known to the online algorithm.}
For an algorithm $\ALG$ and an input instance $\cI = (\CS{M},d,\demands)$, let $\ALG(\cI)$ be the random variable that gets the total cost incurred by the algorithm on $\cI$, and let $\OPT(\cI)$ be the cost of an optimal solution. Also, for a multiset of demand points $\CS{X} \subseteq \demands$, let $\ALG(\CS{X})$ be the random variable that gets the total service cost of the algorithm for the demand points in $\CS{X}$ (i.e., the cost that the algorithm pays at the online rounds when the demand points from $\CS{X}$ arrive).

With a slight abuse of notation, we also use $\OPT(\cI)$ to refer to the optimal solution as a family of clusters, where a cluster is a  multiset of demand points that are assigned to the same facility.\footnote{In case of multiple optimal solutions, we break ties arbitrarily.} We also refer to the facility of a cluster by the name center. When $\mathcal{I}$ is clear from the context, we omit it from the notation and write $\OPT$ instead of $\OPT(\mathcal{I})$.

An algorithm $\ALG$ is called $c$-competitive in the random-order model, if for any input instance $\cI$, $\E{\ALG(\cI)} \leq c \cdot \OPT$, where the expectation is taken over the random arrival order of the demand points, and the internal randomness of $\ALG$.
\section{\texorpdfstring{$\rofl[]$}{DistProb} (Meyerson's Algorithm) is \texorpdfstring{$4$}{4}-Competitive}\label{sec:meyerson}

In this section, we analyze the fundamental randomized online algorithm, described in Algorithm~\ref{alg:rofl}, which was first introduced by Meyerson~\cite{meyerson2001online}. We refer to this algorithm by the name $\rofl[]$.

\IncMargin{1em}
\begin{algorithm2e}[ht]
\caption{$\rofl[]$}\label{alg:rofl}
$F_0 \leftarrow \emptyset$\;
\For {a demand $v_{\ell}$ that arrives at round $\ell$ } {
    $\distalg{v_\ell} \leftarrow d(F_{\ell-1},v_\ell)$\; \label{line:dist}
    $\coin{v_\ell} \leftarrow \min\{\distalg{v_\ell}, 1\}$\;\label{line:coin}
    Flip a coin with probability $\coin{v_\ell}$ of Heads\; \label{line:heads}
    \uIf{Heads}{
        \tcp{Open a facility at $v_\ell$}
        $F_{\ell} \leftarrow F_{\ell-1} \cup \{v_\ell\}$\;
    } \Else {
        $F_{\ell} \leftarrow F_{\ell - 1}$
    }

    Assign $v_\ell$ to its nearest facility in $F_\ell$\;
}
\end{algorithm2e}
\DecMargin{1em}

We bound the cost of $\rofl[]$ for each cluster of $\OPT$ separately. Consider a cluster $\optcluster$ in $\OPT$ with center $\optcenter$. For $u \in \optcluster$, we denote $d^*_u = d(\optcenter,u)$. The cost of $\OPT$ for serving $\optcluster$ is $ \OPT(\optcluster) = 1 + \sum_{u\in \optcluster} d^*_u$.

The basic idea of the analysis is to wait until an online round $T$ in which the algorithm opens a facility at a point $v_T$ from $\optcluster$, and to use this facility to upper bound the distances of the demand points in $\optcluster$ that arrive after round $T$ from their closest open facility (which, in turn, upper bounds the expected cost that the algorithm pays for serving these demand points). 
On the other hand, we also need to upper bound the expected cost that the algorithm pays for serving the demand points in $\optcluster$ that arrive before round $T$. Roughly speaking, to obtain a good upper bound on the distances of the demand points in $\optcluster$ that arrive after time $T$ from their closest open facility, we need the facility $v_T$ to be well placed among these points. However, being too selective about the location of $v_T$ may result in a high service cost for the demand points in $\optcluster$ that arrive before time $T$. Hence, the choice of $T$ should carefully balance between these two considerations.

Before discussing how we define $T$, we establish notations and prove simple lemmas that hold regardless of the definition of $T$. This will help in explaining the intuition behind our definition of $T$.

For a demand point $u\in \demands$ and online round $\ell \in [n]$, let $\distalg{u}$ be the distance between $u$ and the closest open facility at the point in time when $u$ arrives, and $\coin{u} = \min\{\distalg{u},1\}$ ($\distalg{\cdot}$ and $\coin{\cdot}$ are also defined in lines~\ref{line:dist} and \ref{line:coin} of Algorithm~\ref{alg:rofl}). 

Let $T$ be a random variable that gets values in the set $\{1,\dots,n+1\}$. $T$ will get a value of an online round in which a facility from $\optcluster$ is opened. If there is no such online round, $T$ will get the value $n+1$ (for the analysis we define $v_{n+1}$ to be a dummy demand point with an arbitrary value of $d^*_{v_{n+1}}$ that will not be used). As mentioned earlier, the precise definition of $T$ is deferred. 

For an online round $\ell \in [n]$, let $\uptocpoints{\ell} = \{v_1,\dots,v_\ell\} \cap \optcluster$ be the set of demand points from $\optcluster$ that arrive before round $\ell$, and let $\cpoints{\ell}=\optcluster \cap \{v_{\ell},\dots,v_{n}\}$ be the set of remaining demand points in $\optcluster$ at round $\ell$.
We derive upper bounds on the expected service cost of the algorithm for each of the subsets, $\uptocpoints{T}$ and $\cpoints{T+1}$, separately. 

The next lemma shows that to upper bound the cost that $\rofl[]$ pays for the service of the demands in $\uptocpoints{T}$ and $\cpoints{T+1}$, it suffices to upper bound $\E{\sum_{u \in \uptocpoints{T}} \coin{u}}$ and $\E{\sum_{u \in \cpoints{T+1}} \coin{u}}$, respectively.

\begin{lemma}\label{lem:price_factor}
The expected service cost of $\rofl[]$ for a demand $u$ is at most $2 \cdot \E{ \coin{u}}$
\end{lemma}

\begin{proof}
Fix $u \in \demands$. Let $\ALG(u)$ be the cost that $\rofl[]$ pays for serving $u$. Conditioned on $\coin{u} = 1$, the algorithm open a facility at $u$ and pays $1 < 2 = 2 \cdot \coin{u}$ ($1$ for the facility opening cost, and $0$ for the assignment cost). Now let $p \in [0,1)$. Conditioned on $\coin{u} = p$, it holds that $\coin{u} = \distalg{u}$, and the algorithm opens a facility at $u$ with probability $p$ and pays $1$, and with probability $1-p$, it serves $u$ through an open facility at distance $p$ and pays $p$. Hence, $\E{\given{\ALG(u)}{\coin{u} = p}} = p \cdot 1 + (1- p)\cdot p = 2 p - p^2 \leq 2p$. The lemma follows by taking the expectation over $\coin{u}$.
\end{proof} 

In the next lemma, we derive a simple upper bound on the cost that the algorithm pays for the service of the demand points in $\cpoints{T+1}$, as a function of $d^*_{v_T}$. By Lemma~\ref{lem:price_factor} together with the fact that  $\coin{u} = \min\{\distalg{u},1\} \leq \distalg{u}$ for all $u\in \CS{U}$, it suffices to bound $\E{\sum_{u \in \cpoints{T+1}} \distalg{u}}$.

\begin{lemma}\label{lem:C_Tplus1}
$\E{\sum_{u \in \cpoints{T+1}} \distalg{u}} \leq \E{ \sum_{u \in \cpoints{T}} d^*_u} + \E{(|\cpoints{T}| - 2)d^*_{v_T}}$.
\end{lemma}

\begin{proof}
Observe that each demand point $u \in \cpoints{T+1}$ can be served by the open facility at $v_T$, and thus $\E{\distalg{u}} \leq \E{d^*_{v_T} + d^*_u}$. Hence, we get that
\begin{align*}
    \E{\sum_{u \in \cpoints{T+1}} \distalg{u}} &\leq \E{\sum_{u \in \cpoints{T+1}} \left(d^*_{v_T} + d^*_u\right)} 
    = \E{ \left| \cpoints{T+1} \right| d^*_{v_T} + \sum_{u \in \cpoints{T+1}} d^*_u} \\
    &= \E{(|\cpoints{T+1}| - 1)d^*_{v_T} +\sum_{u \in \cpoints{T}} d^*_u } = \E{(|\cpoints{T}| - 2)d^*_{v_T}}  +\E{\sum_{u \in \cpoints{T}} d^*_u },
\end{align*}
where the second equality follows from the fact that $\cpoints{T} = \cpoints{T+1} \cup \{v_T\}$, and also the last equality is due to the fact that $|\cpoints{T+1}| = |\cpoints{T}| - 1$.
\end{proof}

\begin{corollary}\label{cor:C_Tplus1}
$\E{\sum_{u \in \cpoints{T+1}} \coin{u}} \leq \E{ \sum_{u \in \cpoints{T}} d^*_u} + \E{(|\cpoints{T}| - 2)d^*_{v_T}}$.
\end{corollary}

For a demand point $u\in \demands$ and online round $\ell \in [n]$, let $\pul{u}{\ell} = \min\{d(F_{\ell-1}, u),1\}$. Recall that $d(F_{\ell-1},u)$ is the distance between $u$ and the closest facility in $F_{\ell-1}$ (the set of open facilities at the beginning of online round $\ell$). Note that when $u$ arrives at online round $\ell$, i.e., $v_\ell = u$, we have $\coin{u} = \pul{u}{\ell}$, and this is exactly the probability of Heads in line~\ref{line:heads} of Algorithm~\ref{alg:rofl} (that is, the probability that the algorithm opens a facility at $u$).

Next, we proceed with an intuitive, informal discussion about our definition of $T$. As discussed before, we want to choose $T$ in a way that will allow us to obtain good upper bounds on $\E{\sum_{u \in \uptocpoints{T}} \coin{u}}$ and $\E{\sum_{u \in \cpoints{T+1}} \coin{u}}$. Corollary~\ref{cor:C_Tplus1} essentially shows that to get a good upper bound on $\E{\sum_{u \in \cpoints{T+1}} \coin{u}}$, we only need $\E{d^*_{v_T}}$ to be small. Instead of strictly requiring $v_T$ to be ``close'' to $\optcenter$, we apply a less stringent probabilistic approach. We allow the facility $v_T$ to be (sometimes) far away from $\optcenter$, and require only $\E{d^*_{v_T}}$ to be small.

Due to the random arrival order, at each online round $\ell$, each remaining demand point in $\cpoints{\ell}$ is equally likely to \textit{arrive}. Suppose that for every $\ell \in [n]$, each point $u \in \cpoints{\ell}$ had an equal probability to \textit{be opened} at online round $\ell$ (that is, $\pul{u}{\ell} = \pul{u'}{\ell}$ for all $u,u' \in \cpoints{\ell}$). Then, we could simply define $T$ to be the online round in which the first facility from $\optcluster$ is opened by the algorithm. In this hypothetical case, conditioned on the event that a facility $v_T \in \cpoints{T}$ is opened, each point $u \in \cpoints{T}$ was equally likely to be the opened facility $v_T$. Hence, we would have got that the expected distance $d^*_{v_T}$ is the average distance of a point in $\cpoints{T}$ from $\optcenter$, i.e., $\E{d^*_{v_T}} = \E{\frac{1}{|\cpoints{T}|} \sum_{u \in \cpoints{T}} d^*_u}$. Also, with this simple definition of $T$, we have a simple upper bound on $\E{\sum_{u \in \uptocpoints{T}} \coin{u}}$. Using expected waiting time techniques, as used by Meyerson in~\cite{meyerson2001online}, we have $\E{\sum_{u \in \uptocpoints{T}} \coin{u}} \leq 1$.

In reality, however, each point $u \in \cpoints{\ell}$ may have a different probability to be opened by the algorithm. With the simple definition of $T$ as the first round in which a facility from $\optcluster$ is opened, an issue arises when demand points farther away from $\optcenter$ have higher probabilities to be opened than points closer to $\optcenter$, resulting in a bad upper bound on $\E{d^*_{v_T}}$. 

To overcome this issue, we ``balance out" the probabilities by randomly ignoring some openings of facilities in $\optcluster$ which are far from $\optcenter$.
To do so, every time a demand point from $\optcluster$ arrives, we flip an independent coin to decide whether to consider it for the definition of $T$ in case it becomes an open facility. We define $T$ to be the first round in which a facility from $\optcluster$ is opened, and its coin comes up Heads (the probability of Heads is carefully chosen to balance out the facility opening probability). When this happens we say that a \textit{balanced} facility is opened. 
When $v_\ell = u$ the algorithm opens a facility at $u$ with probability $\pul{u}{\ell}$, but we define $T$ to be the current time, and $u$ to be a balanced facility, only with a fraction of this probability.
This fraction is chosen to be no larger than the probabilities of the points that are closer than $u$ to $\optcenter$ to be opened (i.e., the probabilities $\pul{w}{\ell}$ for $w \in \cpoints{\ell}$ with $d^*_{w} \leq d^*_{u}$). More concretely, we take this fraction to be the minimum opening probability of a demand point in $\cpoints{\ell}$ with distance at most $d^*_{u}$ from $\optcenter$ ($\min_{w \in \cpoints{\ell}, d^*_{w} \leq d^*_{u}} p(\ell,w)$). By doing so, we get that conditioned on the event that a balanced facility $v_T \in \cpoints{T}$ is opened, a point $u \in \cpoints{T}$ is at least as likely to be the opened balanced facility $v_T$, as demand points in $\cpoints{T}$ farther away from $\optcenter$. Hence, we get that the expected distance of $d^*_{v_T}$ is at most the average distance of a point in $\cpoints{T}$ from $\optcenter$, i.e., $\E{d^*_{v_T}} \leq \E{\frac{1}{|\cpoints{T}|} \sum_{u \in \cpoints{T}} d^*_u}$.  Lemma~\ref{lem:exp_v_T} gives the formal statement.

To upper bound $\E{\sum_{u \in \uptocpoints{T}} \coin{u}}$, we note that the difference between the actual probability of $u$ to be opened by the algorithm at round $\ell$ (conditioned on $v_\ell = u$), i.e., $\pul{u}{\ell}$, and the probability that $u$ is opened as a balanced facility at round $\ell$, is upper bounded by the difference between the distance of $u$ from its closest open facility in $F_{\ell-1}$, and the distance of some $w \in \cpoints{\ell}$, with $d^*_{w} \leq d^*_{u}$, from its closest open facility, i.e., $d(F_{\ell-1},u) - d(F_{\ell-1},w)$. We will show that this difference is upper bounded by $2d^*_{u}$. Using this fact, we prove in Lemma~\ref{lem:until_opening} that the increase in $\E{\sum_{u \in \uptocpoints{T}} \coin{u}}$ that we incur by waiting for a balanced facility from $\optcluster$ to be opened (instead of any facility in $\optcluster$), is at most $2 \cdot \E{\sum_{u \in \uptocpoints{T} } d^*_u}$.

We now proceed to formalize this intuitive discussion and fill in all the details. We start with some notation and a formal definition of balanced and imbalanced facilities. For an online round $\ell \in [n]$, recall that $\cpoints{\ell}=\optcluster \cap \{v_{\ell},\dots,v_{n}\}$ is the set of remaining demand points in $\optcluster$ at round $\ell$. For $x \in \mathbb{R}_{\geq 0}$, let $\closer{\ell}{x} = \{u \in \cpoints{\ell} : d^*_{u} \leq x\}$, that is, $\closer{\ell}{x}$ is the set of demand points in $\cpoints{\ell}$ whose distance from $\optcenter$ is at most $x$.
When a demand point $v_\ell \in \cpoints{\ell}$ arrives (at online round $\ell$), we consider the demand point 
\begin{align} \label{def:w_v_l}
    w(v_\ell) = \argmin_{u \in \closer{\ell}{ d^*_{v_{\ell}} } } { \pul{u}{\ell} },
\end{align} 
i.e., $w(v_\ell)$ is a point in $\closer{\ell}{ d^*_{v_{\ell}} }$ with a minimal distance to an open facility (in $F_{\ell-1}$).\footnote{Ties are broken arbitrarily.} We illustrate the definition of $w(v_\ell)$ in Figure~\ref{fig:w_def}. The red points are the open facilities in $F_{\ell-1}$ and the blue points are the demand points in $\closer{\ell}{d^*_{v_\ell}}$. The red lines connect each demand point to its closest open facility (regardless of the line pattern). With the solid and dotted line patterns, we highlight two distances of interest. The dotted red line is the shortest distance between a demand point in $\closer{\ell}{d^*_{v_{\ell}}}$, and an open facility in $F_{\ell-1}$, and the solid red line is the distance between $v_\ell$ and its closest open facility in $F_{\ell-1}$. 
\def\r{3}
\def\psize{2}
\def\s2{1.41421356237}
\begin{figure}
    \centering
    \begin{tikzpicture}
    \coordinate (c) at (0,0); 
    \coordinate (v_ell) at (\r / \s2, \r / \s2); 
    \coordinate (u_0) at (-1,1);
    \coordinate (u_1) at (-1,2);
    \coordinate (u_2) at (-2,1);
    \coordinate (u_3) at (1,-1.5);
    \coordinate (u_4) at (2.3,-1.4); 
    \coordinate (u_5) at (-1,-1);
    \coordinate (f_1) at (-4, 4);
    \coordinate (f_2) at (3,-3);
    \coordinate (f_3) at (4.5,3);

    \draw[fill] (c) node[] {} circle (\psize pt) node[below] {$c^*$};
    \draw[fill, color=blue] (v_ell) node[] {} circle (\psize pt) node[above,color=black] {$v_\ell$};
    \draw[color=blue] (0,0) node (circumference) {} circle (\r cm);
    \node[color=blue] at (0,2.6) {$\closer{\ell}{d^*_{v_\ell}}$};
    \draw[<->] (v_ell) -- (c) node [midway, below, yshift=-1mm] {$d^*_{v_\ell}$};

    \draw[fill, color=blue] (u_0) node[] {} circle (\psize pt) node[right, color=black] {};
    \draw[fill, color=blue] (u_1) node[] {} circle (\psize pt) node[right, color=black] {};
    \draw[fill, color=blue] (u_2) node[] {} circle (\psize pt) node[right, color=black] {};
    \draw[fill, color=blue] (u_3) node[] {} circle (\psize pt) node[left, color=black] {};
    \draw[fill, color=blue] (u_4) node[] {} circle (\psize pt) node[above, color=black] {$w(v_\ell)$};
    \draw[fill, color=blue] (u_5) node[] {} circle (\psize pt) node[left, color=black] {};
    
    \draw[fill, color=red] (f_1) node[] {} circle (\psize pt) node[left, color=red] {$F_{\ell-1} \ni\ $};  
    \draw[fill, color=red] (f_2) node[] {} circle (\psize pt) node[right, color=red] {$\ \in F_{\ell-1}$};
    \draw[fill, color=red] (f_3) node[] {} circle (\psize pt) node[right, color=red] {$\ \in F_{\ell-1}$};
    
    \draw[dashed, color=red] (u_3) -- (f_2) node [midway, below, yshift=-1mm] {};
    \draw[thick, dotted, color=red] (u_4) -- (f_2) node [midway, right] {$p(\ell,w(v_\ell))$};    
    \draw[dashed, color=red] (u_5) -- (f_2) node [midway, below, yshift=-1mm] {};  
    
    \draw[dashed, color=red] (u_0) -- (f_1) node [midway, below, yshift=-1mm] {};
    \draw[dashed, color=red] (u_1) -- (f_1) node [midway, below, yshift=-1mm] {};
    \draw[dashed, color=red] (u_2) -- (f_1) node [midway, below, yshift=-1mm] {}; 
    
    \draw[thick, color=red] (v_ell) -- (f_3) node [midway, below] {$p(v_\ell)$};    
    
    \end{tikzpicture}
    \caption{An illustration of $w(v_\ell)$}
    \label{fig:w_def}
\end{figure}
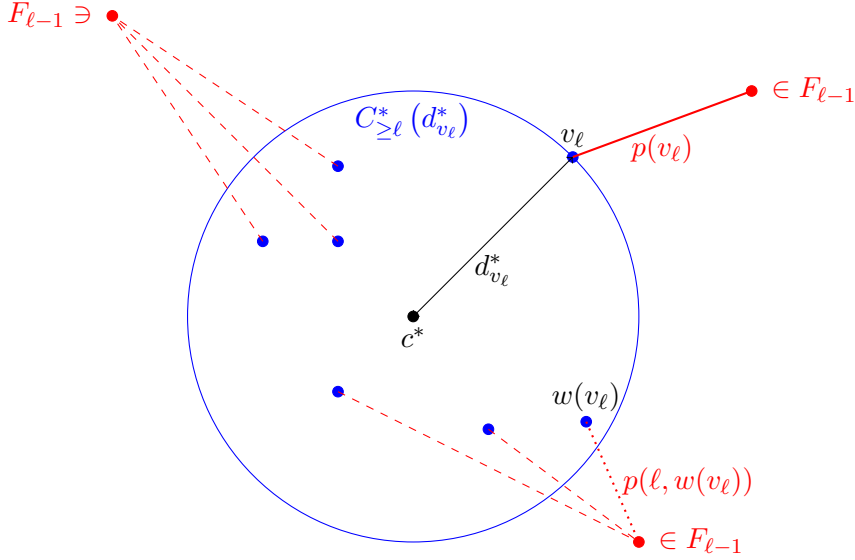

\let\r\relax
\let\psize\relax
\let\s2\relax

Recall that the algorithm flips a coin and opens a facility at $v_\ell$ with probability $\coin{v_\ell}$. We call this coin the \textit{algorithm coin}.
To make our distinction, we flip an additional independent coin, which we call the \textit{analysis coin}, with probability $\pul{w(v_\ell)}{\ell}/ \coin{v_\ell}$ of Heads (note that $\pul{w(v_\ell)}{\ell}/ \coin{v_\ell} \leq 1$). We say that a \textit{balanced} facility is opened at $v_\ell$ if the algorithm coin comes up Heads (a facility is opened at $v_\ell$ by the algorithm) and the analysis coin comes up Heads. If the algorithm coin comes up Heads and the analysis coin comes up Tails, we say that an \textit{imbalanced} facility is opened at $v_\ell$. Observe that overall, a balanced facility is opened at $v_\ell$ if both coins come up Heads which happens with probability $\coin{v_\ell} \cdot \pul{w(v_\ell)}{\ell}/\coin{v_\ell} =  \pul{w(v_\ell)}{\ell}$, and an imbalanced facility is opened at $v_\ell$ with probability $\coin{v_\ell} -  \pul{w(v_\ell)}{\ell}$.

We are now ready to formally define $T$. Let $T$ be the online round in which the first balanced facility from $\optcluster$ is opened (that is, a balanced facility $v_T \in \cpoints{T}$ is opened). If there is no such demand point, we define $T= n+1$ (as mentioned before, for the analysis we define $v_{n+1}$ to be a dummy demand point with an arbitrary value of $d^*_{v_{n+1}}$ that  will not be used).

We continue with Lemma~\ref{lem:exp_v_T} that upper bounds $\E{d^*_{v_T}}$ in terms of the average distance of a point in $\cpoints{T}$ from $\optcenter$, as promised above.

\begin{lemma}\label{lem:exp_v_T}
$\E{|\cpoints{T}|\cdot d^*_{v_T}} \leq \E{\sum_{u \in \cpoints{T}} d^*_{u}}$.
\end{lemma}
\begin{proof}

Let $\ell \in [n + 1]$. We condition on the event $\{T = \ell\}$, i.e., that the first balanced facility from $\optcluster$ is opened at online round $\ell$ (recall that $T=n+1$ means that no balanced facility is opened). Observe that for $T=n+1$, we have $\cpoints{T} = \emptyset$, and the lemma holds.\footnote{For $T=n+1$, we define $d^*_{v_T} = 0$.} 

For $\ell \leq n$, we also condition on the online sequence until round $\ell$, and on the set of open facilities at the beginning of online round $\ell$, i.e., on $\bm{v}_{\ell-1} = (v_1,\dots,v_{\ell-1})$ and $F_{\ell-1}$. Let $\mathbf{x}_{\ell-1} = (x_1,\dots,x_{\ell-1})$ be a sub-sequence of the demand points and let $\CS{F} \subseteq \demands$, such that $\Pr[ T=\ell, \bm{v}_{\ell-1} =\mathbf{x}_{\ell-1}, F_{\ell-1} = \CS{F}]\neq 0$.  For brevity, we denote the event $\{ T=\ell, \bm{v}_{\ell-1} =\mathbf{x}_{\ell-1}, F_{\ell-1} = \CS{F}\}$ by $\event(\mathbf{x}_{\ell-1}, \CS{F})$. 

Note that conditioned on $\event(\mathbf{x}_{\ell-1}, \CS{F}) $, the random variable $\pul{u}{\ell}$ is fixed for all $u \in \demands$, and gets the value  $\pful{\CS{F}}{u}{\ell} = \min\{d(\CS{F}, u), 1\}$. 
Likewise, the set $\cpoints{\ell}$ is also fixed, and gets the value $\CS{Y} = \left( \CS{U} \setminus \{x_1,\dots,x_{\ell-1}\} \right) \cap \optcluster$. Let $u_1,\dots,u_k$ be the demand points in $\CS{Y}$ ordered by their distance from $\optcenter$, i.e., $d^*_{u_1} \leq d^*_{u_2} \leq \dots \leq d^*_{u_k}$. Note that $k = | \CS{Y} | $. For $1\leq  j\leq k$, let $w_j = \argmin_{1 \leq  i \leq j} \{ \pful{\CS{F}}{u_i}{\ell}\}$. Note that $w_j$ is the value of $w(v_\ell)$ (defined in~\eqref{def:w_v_l}) when $u_j$ arrives at round $\ell$, i.e., when $v_\ell = u_j$. We have 
\begin{align}
    \Pr\left[\given{v_T = u_j}{\event(\mathbf{x}_{\ell-1}, \CS{F}) }\right] 
    &= \frac{ \Pr\left[\given{v_T = u_j, T = \ell}{\bm{v}_{\ell-1} =\mathbf{x}_{\ell-1}, 
    F_{\ell-1} = \CS{F}}\right] } { \Pr\left[\given{T = \ell} {\bm{v}_{\ell-1} =\mathbf{x}_{\ell-1}, 
    F_{\ell-1}} = \CS{F}\right] }.\label{eq:conditional_explanation}
\end{align}
Since each remaining demand point in $\demands$ is equally likely to arrive at round $\ell$, it holds that $u_j$ arrives at round $\ell$ with probability $\frac{1}{n-\ell+1}$, and when $v_\ell = u_j$, a balanced facility is opened at $u_j$ with probability $\pful{\CS{F}}{w_j}{\ell}$ (by the definition of balanced opening of a facility). Hence, $\Pr[v_T = u_j, T = \ell \mid \bm{v}_{\ell-1} = \mathbf{x}_{\ell-1}, F_{\ell-1} = \CS{F}] = \frac{1}{n-\ell+1} \cdot \pful{\CS{F}}{w_j}{\ell}$
and $\Pr\left[\given{T = \ell} {\bm{v}_{\ell-1} =\mathbf{x}_{\ell-1}, F_{\ell-1}}\right] = \sum_{i=1}^{k} \frac{1}{n-\ell+1} \cdot \pful{\CS{F}}{w_i}{\ell}$, therefore, by substituting the last expressions in the numerator and denominator of Equation~\eqref{eq:conditional_explanation}, we obtain
\begin{align*}
    \Pr\left[\given{v_T = u_j}{\event(\mathbf{x}_{\ell-1}, \CS{F}) }\right] = \frac{\pful{\CS{F}}{ w_j}{\ell}} {\sum_{i=1}^{k} \pful{\CS{F}}{w_i}{\ell}}.
\end{align*}
Now since $d^{*}_{u_1} \leq \cdots \leq d^*_{u_k}$, and $\pful{\CS{F}}{w_1}{\ell} \geq \pful{\CS{F}}{w_2}{\ell} \geq \cdots \geq \pful{\CS{F}}{w_k}{\ell}$, we get that
\begin{align*}
    \E{\given{|\cpoints{T}| \cdot d^*_{v_T}}{\event(\mathbf{x}_{\ell-1}, \CS{F})}} &= k \cdot \E{\given{ d^*_{v_T}}{\event(\mathbf{x}_{\ell-1}, \CS{F})}} \\ &= k \cdot \sum_{j=1}^{k}  d^*_{u_j}\frac{\pful{\CS{F}}{ w_j}{\ell}} {\sum_{i=1}^{k} \pful{\CS{F}}{w_i}{\ell}} \leq k \cdot \frac{1}{k} \sum_{j=1}^{k}  d^*_{u_j} = \sum_{u \in \CS{Y}}  d^*_{u}.
\end{align*}
To conclude the proof we take the expectation over $T$, $\bm{v}_{T-1}$ and $F_{T-1}$, and get that $\E{|\cpoints{T}| \cdot d^*_{v_T}} \leq \E{\sum_{u \in \CS{C}_T}  d^*_{u} }$.
\end{proof}

We now prove an upper bound for the demand points in $\uptocpoints{T}$. 
\begin{lemma}\label{lem:until_opening}
$\E{\sum_{u \in \uptocpoints{T}} \coin{u}} \leq 1 + 2 \cdot  \E{\sum_{u \in \uptocpoints{T}} d^*_{u}}$
\end{lemma} 

\begin{proof}
For $1\leq \ell \leq n$, let $T_\ell \geq \ell$ be the first online round (from round $\ell$ onward) in which a balanced facility from $\optcluster$ is opened by the algorithm, let $\tcpoints{\ell} = \{v_\ell,\dots,v_{T_\ell}\} \cap \optcluster$, and let $P_\ell = \sum_{u \in \tcpoints{\ell}}\coin{u}$. Observe that $T_1 = T$ and so $\tcpoints{1} = \uptocpoints{T}$. Therefore, to prove the  statement of the lemma we need to show that $\E{ P_1 } \leq 1 + 2 \cdot  \E{\sum_{u \in \tcpoints{1}} d^*_{u}}$.

We condition on the online sequence until round $\ell$, and on the set of open facilities at the beginning of online round $\ell$, i.e., on $\bm{v}_{\ell-1} = (v_1,\dots,v_{\ell-1})$ and $F_{\ell-1}$.
We prove by downwards induction on $\ell$ that for any $\CS{F} \subseteq \demands$ and any sub-sequence of the demand points $\mathbf{x}_{\ell-1} = (x_1,\dots,x_{\ell-1})$ such that $\Pr[\bm{v}_{\ell-1} = \mathbf{x}_{\ell-1} ,F_{\ell-1} = \CS{F}] \neq 0$, it holds that $\E{\given{P_\ell}{\bm{v}_{\ell-1} = \mathbf{x}_{\ell-1} ,F_{\ell-1} = \CS{F}}} \leq 1 + 2  \E{\given{\sum_{u \in \tcpoints{\ell}} d^*_{u}}{ \bm{v}_{\ell-1} = \mathbf{x}_{\ell-1} ,F_{\ell-1} = \CS{F}} }$.

For $\ell = n$, we have $P_n = \coin{v_n}$ if $v_n \in \optcluster$ and $0$ otherwise. Since $\coin{v_n} \leq 1$, we have for any $\mathbf{x}_{n-1}$ and $\CS{F}$ (such that $\Pr[\bm{v}_{n-1} = \mathbf{x}_{n-1}, F_{n-1} = \CS{F}] \neq 0$), that $\E{\given{P_n}{\bm{v}_{n-1} = \mathbf{x}_{n-1}, F_{n-1} = \CS{F}}} \leq 1$.

Now let $\ell < n$. Let $\CS{F} \subseteq \demands$, and let $\mathbf{x}_{\ell-1} = (x_1,\dots,x_{\ell-1})$ be a sub-sequence of the demand points of length $\ell-1$. As in the proof of Lemma~\ref{lem:exp_v_T}, conditioned on $\{F_{\ell-1} = \CS{F}, \bm{v}_{\ell-1} = \mathbf{x}_{\ell-1}\}$, the random variable $\pul{u}{\ell}$ is fixed for all $u \in \demands$, and gets the value $\pful{\CS{F}}{u}{\ell} = \min\{d(\CS{F},u), 1\}$. 
Also $\cpoints{\ell}$ gets the value $\CS{Y} = \left( \CS{U} \setminus \{x_1,\dots,x_{\ell-1}\} \right) \cap \optcluster$. Let $u_1,\dots,u_k$ be the demand points in $\CS{Y}$ ordered by their distance from $\optcenter$, i.e., $d^*_{u_1} \leq d^*_{u_2} \leq \dots \leq d^*_{u_k}$, and for $1\leq  j\leq k$, let $w_j = \argmin_{1 \leq  i \leq j} \{ \pful{\CS{F}}{u_i}{\ell}\}$. As before, note that $w_j$ is the value of $w(v_\ell)$ (defined in~\eqref{def:w_v_l}) when $u_j$ arrives at round $\ell$, i.e., when $v_\ell = u_j$. 

In what follows, we further condition on $v_\ell$, and show that the induction step holds for all possible values of $v_\ell$. We distinguish between two cases, $v_\ell \in \CS{Y}$ and $v_\ell \notin \CS{Y}$. We begin with the case $v_\ell \in \CS{Y}$.

For $u_j \in \CS{Y}$, let $\balanced{u_j}$ be the event that a balanced facility is opened at $u_j$. Likewise, let $\imbalanced{u_j}$ be the event that an imbalanced facility is opened at $u_j$, and let $\nofacility{u_j}$ be the event that no facility is opened at $u_j$. We have
\begin{align}
\begin{split}
    & \E{\given{P_\ell}{	
        \begin{array}{c}
            v_{\ell} = u_j,\\
            \bm{v}_{\ell-1} = \mathbf{x}_{\ell-1},\\
            F_{\ell-1} = \CS{F}
        \end{array}
        }
    }\\ 
    & \quad = \Pr\left[\given{\balanced{u_j}}
        { 
            \begin{array}{c}
                v_{\ell} = u_j,\\
                \bm{v}_{\ell-1} = \mathbf{x}_{\ell-1},\\
                F_{\ell-1} = \CS{F}
            \end{array} 
        } \right] \cdot \pful{\CS{F}}{u_j}{\ell} \\
        & \quad \quad  + \Pr\left[\given{\imbalanced{u_j}}
        {
            \begin{array}{c}
                v_{\ell} = u_j,\\
                \bm{v}_{\ell-1} = \mathbf{x}_{\ell-1},\\
                F_{\ell-1} = \CS{F}
            \end{array}
        } \right] \left( \pful{\CS{F}}{u_j}{\ell} + \E{\given{P_{\ell+1}}{ 	
            \begin{array}{c}
                 v_\ell = u_j, \\
                 \bm{v}_{\ell-1} = \mathbf{x}_{\ell-1}, \\
                 F_{\ell} = \CS{F} \cup \{u_j\}
            \end{array} } } \right) \\
        & \quad \quad + \Pr\left[\given{\nofacility{u_j}}
        {
            \begin{array}{c}
                v_{\ell} = u_j,\\
                \bm{v}_{\ell-1} = \mathbf{x}_{\ell-1},\\
                F_{\ell-1} = \CS{F}
            \end{array}
        } \right] \left( \pful{\CS{F}}{u_j}{\ell} +  \E{\given{P_{\ell+1}}{	
            \begin{array}{c}
                 v_\ell = u_j, \\
                 \bm{v}_{\ell-1} = \mathbf{x}_{\ell-1}, \\
                 F_{\ell} = \CS{F}
            \end{array} }
        } \right).
\end{split}\label{eq:p_l_explained}
\end{align}

By the definition of balanced opening of a facility, the (conditional) probabilities of the events $\balanced{u_j}$, $\imbalanced{u_j}$ and $\nofacility{u_j}$ are $\pful{\CS{F}}{w_j}{\ell}$, $\left(\pful{\CS{F}}{u_j}{\ell} - \pful{\CS{F}}{w_j}{\ell}\right)$ and $\left(1-  \pful{\CS{F}}{u_j}{\ell}\right)$, respectively. By substituting these probabilities in Equation~\eqref{eq:p_l_explained} and using the induction hypothesis (twice) and get that
\begin{align}
\begin{split}
    & \E{\given{P_\ell}{ 
        \begin{array}{c}
            v_{\ell} = u_j,\\
            \bm{v}_{\ell-1} = \mathbf{x}_{\ell-1},\\
            F_{\ell-1} = \CS{F}
        \end{array}
        }
    } \\
    & \quad \leq \pful{\CS{F}}{u_j}{\ell} + \left(\pful{\CS{F}}{u_j}{\ell} - \pful{\CS{F}}{w_j}{\ell}\right) \left(1 + 2 \E{\given{\sum_{u \in \tcpoints{\ell+1}} d^*_{u} }{
            \begin{array}{c}
                 v_\ell = u_j, \\
                 \bm{v}_{\ell-1} = \mathbf{x}_{\ell-1}, \\
                 F_{\ell} = \CS{F} \cup \{u_j\}
            \end{array} } }\right) \\
    &  \quad \quad +\left(1-  \pful{\CS{F}}{u_j}{\ell}\right) \left( 1+ 2\E{\given{\sum_{ u \in \tcpoints{\ell+1}} d^*_{u}}{
            \begin{array}{c}
                 v_\ell = u_j, \\
                 \bm{v}_{\ell-1} = \mathbf{x}_{\ell-1}, \\
                 F_{\ell} = \CS{F}
            \end{array} }
        } \right) \\
    &\quad  = 1 + \left(\pful{\CS{F}}{u_j}{\ell} - \pful{\CS{F}}{w_j}{\ell}\right)  \\
    & \quad \quad + 2\left(\pful{\CS{F}}{u_j}{\ell} - \pful{\CS{F}}{w_j}{\ell}\right)  \E{\given{\sum_{u \in \tcpoints{\ell+1}} d^*_{u} }{
            \begin{array}{c}
                 v_\ell = u_j, \\
                 \bm{v}_{\ell-1} = \mathbf{x}_{\ell-1}, \\
                 F_{\ell} = \CS{F} \cup \{u_j\}
            \end{array} } } \\
    & \quad \quad  + 2\left(1-  \pful{\CS{F}}{u_j}{\ell}\right) \E{\given{\sum_{u \in \tcpoints{\ell+1}} d^*_{u}}{
            \begin{array}{c}
                 v_\ell = u_j, \\
                 \bm{v}_{\ell-1} = \mathbf{x}_{\ell-1}, \\
                 F_{\ell} = \CS{F}
            \end{array} }
        }.
\end{split}\label{eq:exp_round_l_2}
\end{align}

We now upper bound the term $\left(\pful{\CS{F}}{u_j}{\ell} - \pful{\CS{F}}{w_j}{\ell}\right)$. Let $d(\CS{F}, \optcenter)$ be the distance between $\optcenter$ ($\OPT$'s center of $\optcluster$) and the closest open facility in $\CS{F}$. By the triangle inequality, we have $d(\CS{F}, u_j) \leq d^*_{u_j} + d(\CS{F}, \optcenter)$ and $d(\CS{F}, \optcenter) \leq  d(\CS{F}, w_j) + d^*_{w_j} \leq d(\CS{F}, w_j) + d^*_{u_j}$. Hence, 
$d(\CS{F}, u_j) \leq 2d^*_{u_j} + d(\CS{F}, w_j)$, and so 
\begin{align}
    d(\CS{F},u_j) - d(\CS{F}, w_j) \leq  2d^*_{u_j} .\label{eq:triangle_dist}
\end{align}

Now since $d({\CS{F}},u_j) \geq d(\CS{F}, w_j)$, $\pful{\CS{F}}{u_j}{\ell} = \min\{d({\CS{F}},u_j),1\}$ and $\pful{\CS{F}}{w_j}{\ell} = \min\{d({\CS{F}},w_j),1\}$, it holds that $\pful{\CS{F}}{u_j}{\ell} - \pful{\CS{F}}{w_j}{\ell} \leq d(\CS{F},u_j) - d(\CS{F}, w_j) $. To see this, observe that if both $d(\CS{F},u_j) \leq 1$ and $d(\CS{F},w_j) \leq 1$, then  $\pful{\CS{F}}{u_j}{\ell} = d(\CS{F},u_j)$ and $\pful{\CS{F}}{w_j}{\ell} = d(\CS{F},w_j) $. If both $d(\CS{F},u_j)  > 1$ and $d(\CS{F},w_j)  > 1$ then $\pful{\CS{F}}{u_j}{\ell} - \pful{\CS{F}}{w_j}{\ell} = 0 \leq d(\CS{F},u_j) - d(\CS{F}, w_j)$, and if $ d(\CS{F},u_j) > 1$ and $ d(\CS{F},w_j) \leq 1$, then  $\pful{\CS{F}}{u_j}{\ell} - \pful{\CS{F}}{w_j}{\ell} = 1 - d(\CS{F},w_j) <  d(\CS{F},u_j) - d(\CS{F},w_j)$. Hence, we have
\begin{align}
\pful{\CS{F}}{u_j}{\ell} - \pful{\CS{F}}{w_j}{\ell} \leq d(\CS{F},u_j) - d(\CS{F}, w_j) \leq  2d^*_{u_j}. \label{eq:triangle}
\end{align}

Now, observe that 
\begin{align}
\begin{split}
    & \E{\given{\sum_{u \in \tcpoints{\ell}} d^*_{u}} {
            \begin{array}{c}
                 v_\ell = u_j, \\
                 \bm{v}_{\ell-1} = \mathbf{x}_{\ell-1}, \\
                 F_{\ell} = \CS{F}
            \end{array}
    } } \\ 
    & \quad =  \pful{\CS{F}}{w_j}{\ell}  \cdot d^*_{u_j} + \left( \pful{\CS{F}}{u_j}{\ell}  - \pful{\CS{F}}{w_j}{\ell}  \right) 
    \left( 
        d^*_{u_j} +
        \E{ \given{\sum_{u \in \tcpoints{\ell+1}} d^*_{u}} {
                \begin{array}{c}
                 v_\ell = u_j, \\
                 \bm{v}_{\ell-1} = \mathbf{x}_{\ell-1}, \\
                 F_{\ell} = \CS{F} \cup \{u_j\}
            \end{array}
        } } 
    \right) \\
    & \quad \quad + \left(1 -\pful{\CS{F}}{u_j}{\ell}   \right) 
    \left( d^*_{u_j} + \E{ \given{\sum_{u \in \tcpoints{\ell+1}} d^*_{u}} { 
            \begin{array}{c}
                 v_\ell = u_j, \\
                 \bm{v}_{\ell-1} = \mathbf{x}_{\ell-1}, \\
                 F_{\ell} = \CS{F}
            \end{array}
    }} \right) \\
    &\quad = d^*_{u_j} + \left( \pful{\CS{F}}{u_j}{\ell}  - \pful{\CS{F}}{w_j}{\ell}  \right) 
    \E{ \given{\sum_{u \in \tcpoints{\ell+1}} d^*_{u}} {
                \begin{array}{c}
                 v_\ell = u_j, \\
                 \bm{v}_{\ell-1} = \mathbf{x}_{\ell-1}, \\
                 F_{\ell} = \CS{F} \cup \{u_j\}
            \end{array}
        } }   \\
    & \quad \quad + \left(1 - \pful{\CS{F}}{u_j}{\ell}   \right) 
      \E{ \given{\sum_{u \in \tcpoints{\ell+1}} d^*_{u}} { 
            \begin{array}{c}
                 v_\ell = u_j, \\
                 \bm{v}_{\ell-1} = \mathbf{x}_{\ell-1}, \\
                 F_{\ell} = \CS{F}
            \end{array}
    }}.
\end{split}\label{eq:exp_sum_dist}
\end{align}
Thus, by substituting~\eqref{eq:triangle} in~\eqref{eq:exp_round_l_2} and using Equation~\eqref{eq:exp_sum_dist}, we obtain
\begin{align}
\begin{split}
    \E{\given{P_\ell} { 
        \begin{array}{c}
            v_\ell = u_j, \\
            \bm{v}_{\ell-1} = \mathbf{x}_{\ell-1}, \\
            F_{\ell} = \CS{F}
        \end{array}
        } 
    } 
    &\leq  1 + 2  \E{\given{\sum_{u \in \tcpoints{\ell}} d^*_{u}} {        \begin{array}{c}
            v_\ell = u_j, \\
            \bm{v}_{\ell-1} = \mathbf{x}_{\ell-1}, \\
            F_{\ell} = \CS{F}
        \end{array}
        }}.
\end{split}
\end{align}

For $v_\ell = u \in \CS{U} \setminus (\CS{Y} \cup \{x_1,\dots,x_{\ell-1}\} )$, we can consider only whether a facility is opened at $u$ or not. We have
\begin{align}
\begin{split}
    \E{\given{P_\ell}{\begin{array}{c}
                 v_\ell = u, \\
                 \bm{v}_{\ell-1} = \mathbf{x}_{\ell-1}, \\
                 F_{\ell} = \CS{F}
            \end{array}}} 
    &= \pful{\CS{F}}{u}{\ell}
        \E{\given{P_{\ell+1}} { 
            \begin{array}{c}
                 v_\ell = u, \\
                 \bm{v}_{\ell-1} = \mathbf{x}_{\ell-1}, \\
                 F_{\ell} = \CS{F} \cup \{u \}
            \end{array}}} \\
            & \quad + (1-\pful{\CS{F}}{u}{\ell}) 
        \E{\given{P_{\ell+1}}{
            \begin{array}{c}
                 v_\ell = u, \\
                 \bm{v}_{\ell-1} = \mathbf{x}_{\ell-1}, \\
                 F_{\ell} = \CS{F}
            \end{array}
        }} \\
    &\leq 1 + 2 \cdot \pful{\CS{F}}{u}{\ell} 
        \E{ \given{\sum_{u \in \tcpoints{\ell+1}} d^*_{u}} {
            \begin{array}{c}
                 v_\ell = u, \\
                 \bm{v}_{\ell-1} = \mathbf{x}_{\ell-1}, \\
                 F_{\ell} = \CS{F} \cup \{u \}
            \end{array}
        }} \\
    & \quad + 2 \cdot (1-\pful{\CS{F}}{u}{\ell}) \E{ \given{\sum_{u \in \tcpoints{\ell+1}} d^*_{u}} {
            \begin{array}{c}
                 v_\ell = u, \\
                 \bm{v}_{\ell-1} = \mathbf{x}_{\ell-1}, \\
                 F_{\ell} = \CS{F}
            \end{array}
    }} \\
    &\leq 1 + 2 \cdot \E{\given{\sum_{u \in \tcpoints{\ell}} d^*_{u}} {
            \begin{array}{c}
                 v_\ell = u, \\
                 \bm{v}_{\ell-1} = \mathbf{x}_{\ell-1}, \\
                 F_{\ell} = \CS{F}
            \end{array}
    } },
\end{split}\label{eq:not_in_C}
\end{align}
where in the first inequality we used the induction hypothesis, and the last inequality follows by the fact that, similarly to Equation~\eqref{eq:exp_sum_dist}, we have
\begin{align}
\begin{split}
    \E{\given{\sum_{u \in \tcpoints{\ell}} d^*_{u}} {
            \begin{array}{c}
                 v_\ell = u, \\
                 \bm{v}_{\ell-1} = \mathbf{x}_{\ell-1}, \\
                 F_{\ell} = \CS{F}
            \end{array}
    }} 
    &= \pful{\CS{F}}{u}{\ell}  \E{ \given{\sum_{u \in \tcpoints{\ell+1}} d^*_{u}} { 
            \begin{array}{c}
                 v_\ell = u, \\
                 \bm{v}_{\ell-1} = \mathbf{x}_{\ell-1}, \\
                 F_{\ell} = \CS{F} \cup \{u\}
            \end{array}
    }} \\
    & \quad + \left(1 - \pful{\CS{F}}{u}{\ell}  \right)  \E{ \given{\sum_{u \in \tcpoints{\ell+1}} d^*_{u}} {
            \begin{array}{c}
                 v_\ell = u, \\
                 \bm{v}_{\ell-1} = \mathbf{x}_{\ell-1}, \\
                 F_{\ell} = \CS{F}
            \end{array}
    }}.
\end{split}\label{eq:notinc_expression}
\end{align}
To conclude the inductive argument we take the expectation over $v_\ell$.
\end{proof}

We are now ready to prove the competitive-ratio of the algorithm.

\begin{theorem}\label{thm:cr_meyerson}
$\rofl[]$ is $4$-competitive.
\end{theorem}

\begin{proof} 
For a cluster $\optcluster$ in $\OPT$, we have $\optcluster =  \uptocpoints{T} \cup \cpoints{T+1} $. By Lemma~\ref{lem:until_opening} and Corollary~\ref{cor:C_Tplus1} together with Lemma~\ref{lem:price_factor}, we get that
\begin{align*}
\E{\ALG\left(\optcluster\right)} &= \E{\ALG\left(\uptocpoints{T}\right)} +  \E{\ALG\left(\cpoints{T+1}\right)} \\
&\leq 2 \cdot \E{\sum_{u \in \uptocpoints{T}} \coin{u}} + 2 \cdot \E{\sum_{u \in \cpoints{T+1}} \coin{u}} \\
&\leq  2 + 4 \cdot  \E{\sum_{u \in \uptocpoints{T}} d^*_{u}} + 2 \cdot \E{ \sum_{u \in \cpoints{T}} d^*_u} + 2 \cdot \E{(|\cpoints{T}| - 2)d^*_{v_T}} \\
&= 2 + 4 \cdot  \E{\sum_{u \in \uptocpoints{T-1}} d^*_{u}}  + 4 \cdot \E{d^*_{v_T}} + 2 \cdot \E{ \sum_{u \in \cpoints{T}} d^*_u} + 2 \cdot \E{(|\cpoints{T}| - 2)d^*_{v_T}} \\
&= 2 + 4 \cdot  \E{\sum_{u \in \uptocpoints{T-1}} d^*_{u}} + 2 \cdot \E{ \sum_{u \in \cpoints{T}} d^*_u}  + 2 \cdot \E{(|\cpoints{T}|)d^*_{v_T}} \\
&\leq  2 + 4 \cdot  \E{\sum_{u \in \uptocpoints{T-1}} d^*_{u}} + 2 \cdot \E{ \sum_{u \in \cpoints{T}} d^*_u}  + 2 \cdot \E{\sum_{u \in \cpoints{T}} d^*_{u}} = 2 + 4 \cdot  \E{\sum_{u \in \optcluster} d^*_{u}}.
\end{align*}
Where in the last inequality we used Lemma~\ref{lem:exp_v_T}. Now since $\OPT(\optcluster) = 1 + \sum_{u \in \optcluster} d^*_u$, we have $\E{\ALG(\optcluster)} / \OPT(\optcluster) \leq 4$. Since this is true for each cluster $\optcluster$ in $\OPT$, we get that $\E{\ALG}/\OPT \leq 4$.
\end{proof}

We now show that our analysis of $\rofl[]$ is tight.

\begin{theorem}\label{thm:tight_4}
There is an infinite sequence of input instances $\cI_{1},\cI_{2},\dots$, with increasing number of demand points and decreasing distances between the demand points, where the competitive-ratio of $\rofl[]$ on $\cI_{k}$ approaches $4$ as $k \rightarrow \infty$.
\end{theorem}

\begin{proof}
In the instance $\cI_{k}$ the metric space consists of $k+1$ points $\{u_1,\dots,u_k, \optcenter\}$ where $d(u_i,u_j) = 2 \delta$ for all $i\neq j$ and $d(u_i, \optcenter) = \delta$ for all $1 \leq i \leq k$ and $\delta =  1/(4 \sqrt{k})$ (observe that $2 \delta < 1$ for all $k \in \mathbb{N}$). The demand points are  $\demands = \{u_1,\dots,u_k\}$. $\rofl[]$ opens a facility at $v_1$ and pays $1$. Then, for each demand $v_\ell$ that arrives at round $\ell \in \{2,\dots,k\}$ we have $d(F_{\ell-1},v_\ell) = 2\delta$, regardless of the decisions of the algorithm in previous rounds. Hence, the expected cost of the algorithm for serving $v_\ell$ is $2\delta \cdot 1 + (1-2\delta)\cdot 2\delta = 4\delta - 4\delta^2$. Hence, 
\begin{align*}
    \E{\ALG(\cI_{k})} = 1 + (k-1)(4\delta - 4\delta^2).
\end{align*}

On the other hand, $\OPT$ can open a facility at $\optcenter$ and serve each demand by $\optcenter$ at a cost of $\delta$. Thus, $\E{\OPT} \leq 1 + k\delta$. To conclude

\begin{align*}
    \frac{\E{\ALG(\cI_{k})}}{\OPT(\cI_{k})} \geq \frac{1 + (k-1)(4\delta - 4\delta^2)}{1+k\delta} = \frac{1 + (k-1)(1/\sqrt{k} - 1/k)}{1+\sqrt{k}/4},
\end{align*}
which approaches $4$ as $k \rightarrow \infty$.
\end{proof}

\section{Improving \texorpdfstring{$\rofl[]$}{DistProb}}\label{sec:improve}
In this section, we present a modified version of $\rofl[]$ with an improved competitive-ratio. Our construction in Theorem~\ref{thm:tight_4} will be useful to guide us towards the improved algorithm. Observe that on our constructed instances in Theorem~\ref{thm:tight_4}, after the first facility is opened at $v_1$, opening additional facilities at future demand points $v_2,\dots,v_n$ does not reduce the service cost. Nevertheless, $\rofl[]$ randomly opens a facility at each arriving demand point $v_\ell$ with probability $\coin{v_\ell} = 2 \delta$, which leads to an expected service cost of $4\delta- 4\delta^2$, while $\OPT$ pays only $\delta$. In this way, we get the competitive-ratio of $4$. On these instances, reducing the probability of opening a facility results in better performance.

Generally, for any $\ell \in [n]$ we modify the probability of opening a facility at a demand point $v_\ell$ at distance $\distalg{v_\ell} = d(F_{\ell-1}, v_\ell)$ and reduce it from $\min\{\distalg{v_\ell},1\}$ to some probability $g(\distalg{v_\ell})$, for $g : \mathbb{R}_{\geq 0} \rightarrow [0,1]$ (see Algorithm~\ref{alg:rofl_q} for a formal description). Observe that for $\distalg{v_\ell} \geq 1$, it is always better to open a facility (with probability $1$), and for $\distalg{v_\ell} = 0$ there is no reason to open an additional facility. Therefore, we can focus our attention on $\distalg{v_\ell} \in (0,1)$.

\IncMargin{1em}
\begin{algorithm2e}[ht]
\caption{Generic $\rofl[]$}\label{alg:rofl_q}
$F_0 \leftarrow \emptyset$\;
\For {a demand $v_{\ell}$ that arrives at round $\ell$ } {
    $\distalg{v_\ell} \leftarrow d(F_{\ell-1}, v_\ell)$\;
    $\coin{v_\ell} \leftarrow g(\distalg{v_\ell})$\;\label{line:flip_start}
    Flip a coin with probability $\coin{v_\ell}$ of Heads\; \label{line:roflq_heads}
    \uIf{Heads}{
        $F_{\ell} \leftarrow F_{\ell-1} \cup \{v_\ell\}$\;
    } \Else {
        $F_{\ell} \leftarrow F_{\ell - 1}$
    }\label{line:flip_end}

    Assign $v_\ell$ to the nearest facility in $F_\ell$\;
}
\end{algorithm2e}
\DecMargin{1em}

By our observation above, to improve upon the competitive-ratio of $4$, $g$ has to satisfy $g(\distletter) < \distletter$. Next, we construct another example to derive a lower bound on $g(\distletter)$. Concretely, we show that to get an improvement, $g$ must satisfy $g(\distletter) > \distletter/4$. These two bounds leads us to the choice of a function $g$ of the form $g(\distletter) = q \cdot \distletter$ for some $q \in (1/4,1)$.

We now derive the lower bound on $g(\distletter)$. The idea is simple: We construct a family of instances in which a very large number of demand points arrive at each point in the metric space. This way, the best approach is to open a facility at each point in the metric space. Hence, on these instances, larger $g(\distletter)$ provides better performance. 

\begin{theorem}\label{thm:f_lower_bound}
For all $\delta \in (0,1)$, there is an infinite sequence of input instances $\cI_{\delta,1},\cI_{\delta,2},\dots$, in metric spaces with uniform distances of $\delta$, where the competitive-ratio of Algorithm~\ref{alg:rofl_q} on $\cI_{\delta,k}$ approaches $1 + \delta/g(\delta) - \delta$ as $k \rightarrow \infty$.
\end{theorem}

\begin{proof}
Let $\delta \in (0,1)$. In the instance $\cI_{\delta, k}$ the metric space consists of $k$ points $\{c_1,\dots,c_k\}$ where $d(c_i,c_j) = \delta$ for all $1 \leq i < j \leq k$. There are $k^2$ demand points in total, $k$ demand points arrive at each location $c_i$ in the metric space. 
First, observe that $\OPT \leq k$ as one can open a facility at each point in the metric space and pay $k$. On the other hand, the algorithm first opens a facility at the location $c_{t_1}$ of the first arriving demand point $v_1$ and pays $1$. Let $c_{t_2},\dots, c_{t_k}$ be the remaining  points in the metric space. Before a facility is opened at $c_{t_j}$, the expected cost that the algorithm pays for each demand point that arrive at $c_{t_j}$ is $g(\delta) \cdot 1 + (1-g(\delta)) \cdot \delta$. Hence, the expected cost that the algorithm pays for serving the demand points at $c_{t_j}$ is $
    \sum_{i=1}^{k} \left( 1 - g(\delta) \right)^{i-1} \left( g(\delta) + (1-g(\delta)) \delta \right) = \left( g(\delta) + (1-g(\delta)) \delta \right) \left( 1- (1-g(\delta))^{k} \right) /g(\delta)$.
Therefore, we get that the competitive-ratio of the algorithm is lower bounded by
\begin{align*}
    \frac{1 + (k-1) \left( g(\delta) + (1-g(\delta)) \delta \right) \left(1- (1-g(\delta))^{k}\right)/ g(\delta)}
    {k},
\end{align*}
which approaches $1 + \delta/g(\delta) - \delta$, as $k$ approaches infinity. 
\end{proof}

Following Theorem~\ref{thm:f_lower_bound}, to improve upon a competitive-ratio of $4$, $g$ must satisfy $1 + \delta/g(\delta) - \delta < 4$ and so $g(\delta) > \delta/(3+\delta) > \delta/4$, for all $\delta \in (0,1)$.

To sum up, we choose $g$ of the form $g(\distletter) = q \cdot \distletter$. We note that to simplify our analysis, we use the continuous function $g(\distletter) = \min\{ q\cdot \distletter, 1\}$, instead of the piecewise function $g(\distletter) = q \cdot \distletter$, if $\distletter \leq 1$, and $g(\distletter) = 1 $, otherwise. This choice has no impact on the competitive-ratio, and the same techniques can be used to analyze the piecewise function and gives the same results. Yet, in practice, it is always better to use the peicewise function. We refer to Algorithm~\ref{alg:rofl_q} with the choice $g(\distletter) = \min\{ q\cdot \distletter, 1\}$ by the name $\rofl$.

From our constructed instances above with the choice $g(\distletter) = q \cdot \distletter$, we get that the competitive-ratio of $\rofl$ is at most $1 + 1/q - \delta$, and since we can choose an arbitrarily small $\delta$, we get the following.

\begin{theorem}
There is an infinite sequence of instances $\cI_{1},\cI_{2},\dots$, where the competitive-ratio of $\rofl$ on $\cI_{k}$ approaches $1+1/q$ as $k \rightarrow \infty$.
\end{theorem}

Also, with the instances from the proof of Theorem~\ref{thm:tight_4}, we get the following result.

\begin{theorem}
There is an infinite sequence of instances $\cI_{1},\cI_{2},\dots$, where the competitive-ratio of $\rofl$ on $\cI_{k}$ approaches $2(1+q)$ as $k \rightarrow \infty$.
\end{theorem}

Together, we have the following corollary.
\begin{corollary}\label{cor:roflq_lower_bound}
The competitive-ratio of $\rofl$ is at least $(1+q)\max\{2,1/q\} - o(1)$.
\end{corollary}

We now move to analyze the performance of $\rofl$.
We use the same notations from Section~\ref{sec:meyerson}, and redefine $\pul{u}{\ell}$ to be compatible with $\rofl$: For a demand $u\in \demands$ and online round $\ell \in [n]$, let $\pul{u}{\ell} = \min\{q \cdot d(F_{\ell-1}, u), 1\}$. Note that $\pul{u}{\ell}$ is exactly the probability of Heads in line~\ref{line:roflq_heads} of Algorithm~\ref{alg:rofl_q} with $g(\distletter) = \min\{q\cdot \distletter,1\}$ when $v_\ell = u$ (also note that for $q = 1$, $\pul{u}{\ell}$ coincides with our original definition for $\rofl[]$). Except for the definition of $\pul{u}{\ell}$, our definition of the analysis coin as well as the distinction between a balanced and imbalanced opening of a facility remain the same as in Section~\ref{sec:meyerson}.

Our analysis for $\rofl$ is similar to our analysis of $\rofl[]$ in Section~\ref{sec:meyerson}, therefore, we refer to our analysis of $\rofl[]$ when the details remain the same and apply to $\rofl$, and prove analogues lemmas for $\rofl$ when it is required.

We bound the cost of $\rofl$ on each cluster of $\OPT$ separately. Consider a cluster $\optcluster$ in $\OPT$ with center $\optcenter$. Recall that for $\ell \in [n]$, we define $\cpoints{\ell} \subseteq \optcluster$ to be the set of the remaining demand points $u \in \optcluster$ at round $\ell$, and that $T$ is the online round in which the first balanced facility from $\optcluster$ is opened.  We partition $\optcluster$ in exactly the same way we did in the analysis of $\rofl[]$, that is, $\optcluster =  \uptocpoints{T} \cup \cpoints{T+1}$.

We begin by proving an analogue of Lemma~\ref{lem:price_factor}.
\begin{lemma}\label{lem:price_factor_roflq}
The expected cost of $\rofl$ on a demand $u$ is at most $(1+1/q) \cdot \E{\coin{u}}$
\end{lemma}

\begin{proof}
Fix $u \in \demands$. Let $\ALG(u)$ be the cost that $\rofl$ pays for serving $u$. Conditioned on $\coin{u} = 1$, the algorithm open a facility at $u$ and pays $1 \leq (1+1/q)$. Now let $p\in [0,1)$. Conditioned on $\coin{u} = p < 1$, it holds that $\coin{u} = q \cdot \distalg{u}$ and $\distalg{u} = p/q$. $\rofl$ opens a facility at $u$ with probability $p$ and pays $1$, and with probability $(1-p)$, it serves $u$ through an open facility at distance $p/q$ and pays $p/q$. Hence, $\E{\given{\ALG(u)}{\coin{u} = p}} = p \cdot 1 + (1- p) p/q = p + p/q - p^2/q \leq (1+1/q)p$. The lemma follows by taking the expectation over $\coin{u}$.
\end{proof}

To bound the cost of $\rofl$ on $\cpoints{T+1}$, we note that Lemma~\ref{lem:C_Tplus1} and Lemma~\ref{lem:exp_v_T} from Section~\ref{sec:meyerson} also apply to $\rofl$. With the new definition of $\coin{u}$, we have for all $u \in \CS{U}$ that $\coin{u} \leq q \cdot \distalg{u}$, so we get the following corollary of Lemma~\ref{lem:C_Tplus1} (analogously to Corollary~\ref{cor:C_Tplus1}).
\begin{corollary}\label{cor:C_Tplus1_roflq}
$\E{\sum_{u \in \cpoints{T+1}} \coin{u}} \leq q \cdot \E{ \sum_{u \in \cpoints{T}} d^*_u} + q \cdot \E{(|\cpoints{T}| - 2)d^*_{v_T}}$.
\end{corollary}

It remains to derive bounds for the demand points in $\uptocpoints{T}$. To this end, we prove an analogue of Lemma~\ref{lem:until_opening}.

\begin{lemma}\label{lem:roflq_until_opening}
$\E{\sum_{u \in \uptocpoints{T}} \coin{u}} \leq 1 + 2q \cdot  \E{\sum_{u \in \uptocpoints{T}} d^*_{u}}$
\end{lemma} 

\begin{proof}
The proof is very similar to the proof of Lemma~\ref{lem:until_opening}. We use the same definitions and notations as in the proof of Lemma~\ref{lem:until_opening}, except for the new definition of $\pful{\CS{F}}{u}{\ell} = \min\{q \cdot d(\CS{F},u), 1\}$ (instead of $\min\{d(\CS{F},u), 1\}$). With this value of $\pful{\CS{F}}{u}{\ell}$, we prove a more general inductive statement.

Recall that for $1\leq \ell \leq n$, we define $T_\ell \geq \ell$ to be the first online round (from round $\ell$ onward) in which a balanced facility is opened by the algorithm, $\tcpoints{\ell} = \{v_\ell,\dots,v_{T_\ell}\} \cap \optcluster$, and $P_\ell = \sum_{u \in \tcpoints{\ell}} \coin{u}$. We prove by downwards induction on $\ell$ that for any $\CS{F} \subseteq \demands$ and any sub-sequence of the demand points $\mathbf{x}_{\ell-1} = (x_1,\dots,x_{\ell-1})$ such that $\Pr[\bm{v}_{\ell-1} = \mathbf{x}_{\ell-1} ,F_{\ell-1} = \CS{F}] \neq 0$, it holds that $\E{\given{P_\ell}{\bm{v}_{\ell-1} = \mathbf{x}_{\ell-1} ,F_{\ell-1} = \CS{F}}} \leq 1+ 2q\E{\given{\sum_{u \in \tcpoints{\ell}} d^*_{u}}{ \bm{v}_{\ell-1} = \mathbf{x}_{\ell-1} ,F_{\ell-1} = \CS{F}} } $.

For $\ell = n$, we have
$P_n = \coin{v_n}$ if $v_n \in \optcluster$ and $0$ otherwise. Since $\coin{v_n} \leq 1$, the base case of the induction holds. For $\ell < n$, the proof proceeds as in the proof of Lemma~\ref{lem:until_opening}, until we get to  Equation~\eqref{eq:p_l_explained}. Now, similarly to Inequality~\eqref{eq:exp_round_l_2}, we substitute the probabilities with their respective values and use the new induction hypothesis to obtain that
\begin{align}
\begin{split}
    & \E{\given{P_\ell}{ 
        \begin{array}{c}
            v_{\ell} = u_j,\\
            \bm{v}_{\ell-1} = \mathbf{x}_{\ell-1},\\
            F_{\ell-1} = \CS{F}
        \end{array}
        }
    } \\
    & \quad \leq \pful{\CS{F}}{u_j}{\ell} + \left(\pful{\CS{F}}{u_j}{\ell} - \pful{\CS{F}}{w_j}{\ell}\right) \left(1 + 2 q \E{\given{\sum_{u \in \tcpoints{\ell+1}} d^*_{u} }{
            \begin{array}{c}
                 v_\ell = u_j, \\
                 \bm{v}_{\ell-1} = \mathbf{x}_{\ell-1}, \\
                 F_{\ell} = \CS{F} \cup \{u_j\}
            \end{array} } }\right) \\
    & \quad \quad +\left(1-  \pful{\CS{F}}{u_j}{\ell}\right) \left( 1+ 2q \E{\given{\sum_{u \in \tcpoints{\ell+1}} d^*_{u}}{
            \begin{array}{c}
                 v_\ell = u_j, \\
                 \bm{v}_{\ell-1} = \mathbf{x}_{\ell-1}, \\
                 F_{\ell} = \CS{F}
            \end{array} }
        } \right) \\
    & \quad = 1 + \left(\pful{\CS{F}}{u_j}{\ell} - \pful{\CS{F}}{w_j}{\ell}\right)  \\
    & \quad \quad + 2q\left(\pful{\CS{F}}{u_j}{\ell} - \pful{\CS{F}}{w_j}{\ell}\right)  \E{\given{\sum_{u \in \tcpoints{\ell+1}} d^*_{u} }{
            \begin{array}{c}
                 v_\ell = u_j, \\
                 \bm{v}_{\ell-1} = \mathbf{x}_{\ell-1}, \\
                 F_{\ell} = \CS{F} \cup \{u_j\}
            \end{array} } } \\
    & \quad \quad   + 2q\left(1-  \pful{\CS{F}}{u_j}{\ell}\right) \E{\given{\sum_{u \in \tcpoints{\ell+1}} d^*_{u}}{
            \begin{array}{c}
                 v_\ell = u_j, \\
                 \bm{v}_{\ell-1} = \mathbf{x}_{\ell-1}, \\
                 F_{\ell} = \CS{F}
            \end{array} }
        }.
\end{split}\label{eq:step2_roflq}
\end{align}

We now upper bound the term $\left(\pful{\CS{F}}{u_j}{\ell} - \pful{\CS{F}}{w_j}{\ell}\right)$. To this end we use Inequality~\eqref{eq:triangle_dist} (which still holds), and show that $\pful{\CS{F}}{u_j}{\ell} - \pful{\CS{F}}{w_j}{\ell} \leq q \cdot ( d(\CS{F}, u_j) - d(\CS{F}, w_j))$. Recall that $d(\CS{F}, u_j) \geq d(\CS{F}, w_j)$, $\pful{\CS{F}}{u_j}{\ell} = \min\{q\cdot d(\CS{F}, u_j),1\}$ and $\pful{\CS{F}}{w_j}{\ell} = \min\{q \cdot d(\CS{F}, w_j),1\}$. We distinguish between the following cases: If both $q \cdot d(\CS{F}, u_j) \leq 1$ and $q \cdot d(\CS{F}, w_j) \leq 1$, then  $\pful{\CS{F}}{u_j}{\ell} = q \cdot d(\CS{F}, u_j)$ and $\pful{\CS{F}}{w_j}{\ell} = q \cdot d(\CS{F}, w_j)$ and the claim holds. If both $q \cdot d(\CS{F}, u_j) > 1$ and $q \cdot d(\CS{F}, w_j) > 1$ then $\pful{\CS{F}}{u_j}{\ell} - \pful{\CS{F}}{w_j}{\ell} = 0 \leq q \cdot d(\CS{F}, u_j) - q \cdot d(\CS{F}, w_j)$. Finally, if $q \cdot d(\CS{F}, u_j) > 1$ and $q \cdot  d(\CS{F}, w_j) \leq 1$, then  $\pful{\CS{F}}{u_j}{\ell} - \pful{\CS{F}}{w_j}{\ell} = 1 - q \cdot d(\CS{F}, w_j) <  q \cdot d(\CS{F}, u_j) - q \cdot d(\CS{F}, w_j)$. Hence, we obtain 
\begin{align}
    \pful{\CS{F}}{u_j}{\ell} - \pful{\CS{F}}{w_j}{\ell} \leq  2q  d^*_{u_j}.\label{eq:triangle_roflq}
\end{align}
To conclude, we substitute Inequality~\eqref{eq:triangle_roflq} and use Equation~\eqref{eq:exp_sum_dist} (which still holds) in Inequality~\eqref{eq:step2_roflq}, and get that 
\begin{align}
\begin{split}
    \E{\given{P_\ell} { 
        \begin{array}{c}
            v_\ell = u_j, \\
            \bm{v}_{\ell-1} = \mathbf{x}_{\ell-1}, \\
            F_{\ell} = \CS{F}
        \end{array}
        } 
    } 
    &\leq 1 + 2q  \E{\given{\sum_{u \in \tcpoints{\ell}} d^*_{u}} {        \begin{array}{c}
            v_\ell = u_j, \\
            \bm{v}_{\ell-1} = \mathbf{x}_{\ell-1}, \\
            F_{\ell} = \CS{F}
        \end{array}
        }}.
\end{split}
\end{align}

For $v_\ell = u \in \CS{U} \setminus \CS{Y}$, similarly to Inequality~\eqref{eq:not_in_C} (with our new induction hypothesis) we also get that
\begin{align*}
\begin{split}
    \E{\given{P_\ell}{\begin{array}{c}
                 v_\ell = u, \\
                 \bm{v}_{\ell-1} = \mathbf{x}_{\ell-1}, \\
                 F_{\ell} = \CS{F}
            \end{array}}} 
    &= p_{\CS{F}}(\ell, u) 
        \E{\given{P_{\ell+1}} { 
            \begin{array}{c}
                 v_\ell = u, \\
                 \bm{v}_{\ell-1} = \mathbf{x}_{\ell-1}, \\
                 F_{\ell} = \CS{F} \cup \{u \}
            \end{array}}} 
        + (1-p_{\CS{F}}(\ell, u)) 
        \E{\given{P_{\ell+1}}{
            \begin{array}{c}
                 v_\ell = u, \\
                 \bm{v}_{\ell-1} = \mathbf{x}_{\ell-1}, \\
                 F_{\ell} = \CS{F}
            \end{array}
        }} \\
    &\leq 1 + 2 q \cdot p_{\CS{F}}(\ell, u) 
        \E{ \given{\sum_{u \in \tcpoints{\ell+1}} d^*_{u}} {
            \begin{array}{c}
                 v_\ell = u, \\
                 \bm{v}_{\ell-1} = \mathbf{x}_{\ell-1}, \\
                 F_{\ell} = \CS{F} \cup \{u \}
            \end{array}
        }} \\
    & \quad + 2  q\cdot (1-p_{\CS{F}}(\ell, u)) \E{ \given{\sum_{u \in \tcpoints{\ell+1}} d^*_{u}} {
            \begin{array}{c}
                 v_\ell = u, \\
                 \bm{v}_{\ell-1} = \mathbf{x}_{\ell-1}, \\
                 F_{\ell} = \CS{F}
            \end{array}
    }} \\
    &\leq 1 + 2 q\cdot \E{\given{\sum_{u \in \tcpoints{\ell}} d^*_{u}} {
            \begin{array}{c}
                 v_\ell = u, \\
                 \bm{v}_{\ell-1} = \mathbf{x}_{\ell-1}, \\
                 F_{\ell} = \CS{F}
            \end{array}
    } },
\end{split}
\end{align*}
where in the last inequality we used Equation~\eqref{eq:notinc_expression}.
To conclude the inductive argument we take the expectation over $v_\ell$.
\end{proof}

We can now put all the pieces together and derive the competitive-ratio of $\rofl$.

\begin{theorem}\label{thm:qrofl-competitive}
$\rofl$ is $(1+q)\max\{2,1/q\}$-competitive.
\end{theorem}
\begin{proof} 
Similarly to the proof of Theorem~\ref{thm:cr_meyerson},
for a cluster $\optcluster$ in $\OPT$, we have $\optcluster =  \uptocpoints{T} \cup \cpoints{T+1} $. By Corollary~\ref{cor:C_Tplus1_roflq} and Lemma~\ref{lem:roflq_until_opening} together with Lemma~\ref{lem:price_factor_roflq}, that
\begin{align*}
\E{\ALG\left(\optcluster\right)} &= \E{\ALG\left(\uptocpoints{T}\right)} +  \E{\ALG\left(\cpoints{T+1}\right)} \\
&\leq (1+1/q)  \E{\sum_{u \in \uptocpoints{T}} \coin{u}} + (1+1/q)  \E{\sum_{u \in \cpoints{T+1}} \coin{u}} \\
&\leq  1+1/q + 2(1+q)\E{\sum_{u \in \uptocpoints{T}} d^*_{u}} +  (1+q) \E{ \sum_{u \in \cpoints{T}} d^*_u} +  (1+q) \E{(|\cpoints{T}| - 2)d^*_{v_T}} \\
&= 1+1/q + 2(1+q)   \E{\sum_{u \in \uptocpoints{T-1}} d^*_{u}}  + 2(1+q)  \E{d^*_{v_T}} + (1+q)  \E{ \sum_{u \in \cpoints{T}} d^*_u}  \\
&\quad + (1+q)  \E{(|\cpoints{T}| - 2)d^*_{v_T}} \\
&= 1+1/q + 2(1+q)   \E{\sum_{u \in \uptocpoints{T-1}} d^*_{u}} + (1+q)  \E{ \sum_{u \in \cpoints{T}} d^*_u}  + (1+q) \E{(|\cpoints{T}|)d^*_{v_T}} \\
&\leq  1+1/q + 2(1+q) \E{\sum_{u \in \uptocpoints{T-1}} d^*_{u}} + (1+q)\E{ \sum_{u \in \cpoints{T}} d^*_u}  + (1+q) \E{\sum_{u \in \cpoints{T}} d^*_{u}} \\
&= 1+1/q + 2(1+q)  \E{\sum_{u \in \optcluster} d^*_{u}}.
\end{align*}
Where in the last inequality we used Lemma~\ref{lem:exp_v_T}. Now since $\OPT(\optcluster) = 1 + \sum_{u \in \optcluster} d^*_u$, we have $\E{\ALG(\optcluster)} / \OPT(\optcluster) \leq \max\{1+1/q,2(1+q)\}$. Since this is true for each cluster $\optcluster$ in $\OPT$, we get that $\E{\ALG}/\OPT \leq (1+q)\max\{1/q,2\}$.
\end{proof}

By Corollary~\ref{cor:roflq_lower_bound}, we get that our analysis of $\rofl$ is tight for all $q \in (0,1)$. Optimizing over the choice of $q$, the best competitive-ratio of $\rofl$ is obtained for $q=1/2$, for which we get a competitive-ratio of $3$.

\begin{corollary}
$\rofl[\nicefrac{1}{2}]$ is $3$-competitive.
\end{corollary}
\section{Lower Bound}\label{sec:lowerbound}

In this section, we show that no online algorithm can have a competitive-ratio better than $2$, even in the weaker online $\IID$ model with full prior knowledge of the distribution. 

In the $\IID$ model, we are given the metric space $(\CS{M},d)$, and a distribution $D$ over $\CS{M}$ upfront. Then, at each online round $\ell \in [n]$, the demand point $v_\ell$ is drawn independently from $D$. In this model, an algorithm $\ALG$ is called $c$-competitive, if for any input instance $\cI = (\CS{M},d,D,n)$, it holds that $\E{\ALG(\cI)} \leq c \cdot \E{\OPT(\cI)}$, where the expectation is taken over $v_1,\dots,v_n \sim D$, and the internal randomness of the algorithm.

\begin{theorem}\label{thm:lower_bound}
Let $\ALG$ be an algorithm for online facility location in the $\IID$ model, then, the competitive-ratio of $\ALG$ is at least $2 - o(1)$.
\end{theorem}

\begin{proof}
Let $m = n^2$. We construct a metric space with $m + \binom{m}{n}$ points. The metric space consists of two types of points: The first type consists of $m$ points, $x_1,\dots,x_m$, with $d(x_i,x_j) = 1$ for all $i,j\in [m]$. The second type of points are called \textit{subset points}. For each subset $\CS{I} \subseteq [m]$ of cardinality $n$, there is a point $s_{\CS{I}}$ with $d(s_\CS{I},x_j) = 1/2$ if $j \in \CS{I}$ and $d(s_\CS{I},x_j) = 1$ otherwise. Finally, for two subset points $s_\CS{I} \neq s_\CS{J}$, $d(s_\CS{I},s_\CS{J}) = 1$. For the distribution $D$, we take the uniform distribution over $\{x_1,\dots,x_m\}$. 

To upper bound the cost of $\OPT$, observe that the set of arriving demand points $\{v_1,\dots,v_n\}$ is a subset of the points in $\{x_1,\dots,x_m\}$ of cardinality at most $n$, and therefore, there is a subset point $s_{\CS{I}}$ at distance at most $1/2$ from all the arriving demand points. Hence, $\E{\OPT} \leq 1 + n/2$.

We now consider the performance of $\ALG$. When a demand point $v_\ell$ arrives at $x_i$, we distinguish between three cases. First, if there is an open facility at $x_i$, the algorithm can serve the demand point at no cost. Second, if there is an open facility at a subset point $s_\CS{I}$ for $i \in \CS{I}$, the algorithm can assign the demand point to $s_{\CS{I}}$ at a cost of $1/2$. Otherwise, it must pay at least $1$ for serving $v_\ell$ (either by opening a facility at $v_\ell$ or by assigning it to an open facility).


To lower bound the cost of $\ALG$, we start by charging $\ALG$ a cost of $1$ for each demand point. Then, we subtract the cost saved by $\ALG$ due to demand points that arrive at the same location, and due to the opening of subset facilities. At online round $\ell \in [n]$, the probability that $v_\ell$ arrives at the same location as one of the previous demand points $\{v_1,\dots,v_{\ell-1}\}$, is at most $(\ell - 1)/m < n/m$. Hence, the expected cost saved by $\ALG$ for serving demand points that arrive at the same location is at most $n^2/m = 1$.

For the cost saved by subset facilities, when the algorithm opens a subset facility $s_{\CS{I}}$ at online round $\ell$, it pays an opening cost of $1$, and saves a cost of at most $1/2$ for the assignment cost of $v_\ell$ (the demand point at round $\ell$). Then, at each successive online round $j \in \{ \ell+1,\dots,n\}$, $s_{\CS{I}}$ saves a cost of at most $1/2$ if $v_j$ arrives at some $x_i$ for $i \in \CS{I}$, which happens with probability $n/m$. Hence, the expected cost saved by opening a facility at $s_\CS{I}$ is at most $1/2 + \frac{n^2}{2m} - 1 = \frac{n^2}{2m} - 1/2$ (note that we subtract $1$ to account for the facility opening cost). Since $m = n^2$, we have $\frac{n^2}{2m} - 1/2 = 0$, and so, the expected cost of $\ALG$ is not reduced by opening subset facilities.

We get that $\E{\ALG} \geq n - 1$. And so, $\E{\ALG}/\E{\OPT} \geq (n-1)/(1+n/2) = 2 - 6/(n+2)$, which approaches $2$ as $n$ approaches infinity.
\end{proof}

A simple and well-known result is that any $c$-competitive online algorithm in the random-order model, is also $c$-competitive in the $\IID$ model (see e.g.,~\cite{gupta2022random,mehta2013online}). Hence, we have the following corollary.

\begin{corollary}
Let $\ALG$ be an algorithm for online facility location, then, the random-order competitive-ratio of $\ALG$ is at least $2 - o(1)$.
\end{corollary}
\section{Mixed Adversarial and Random Arrival Order}\label{sec:mixed}

Interestingly, Meyerson's analysis of $\rofl[]$ in~\cite{meyerson2001online} does not fully utilize the random arrival order of the demand points. More concretely, in~\cite{meyerson2001online} the demand points in each cluster of $\OPT$ are partitioned into  ``close'' and ``far'' points. The close points are half of the demand points which are closest to the optimal center. The analysis in~\cite{meyerson2001online} proves that $\rofl[]$ is $8$-competitive, regardless of the relative ordering of the close points in each cluster (amongst themselves). It only requires that each far point arrives in a uniformly random position between the close points. This result indicates that $\rofl[]$ is robust to adversarial ordering of the close points within each cluster.

Since our analysis does not rely on Meyerson's partition into close and far points, we can prove a stronger statement regarding the robustness of $\rofl$ to partial adversarial orders. We show that for a parameter $\rho \in (0,1)$, with an additional cost of a factor of at most $(2-\rho)/\rho$ in the competitive-ratio of $\rofl$, our analysis holds even when an arbitrary $(1-\rho)$-fraction of the demand points in each cluster are ordered adversarially, and the remaining demand points are randomly positioned between them.

Formally, for $\rho \in (0,1)$, in the \textit{$\rho$-partial random-order} setting, the following process generates the online sequence:
\begin{enumerate}
\item The adversary chooses an input instance, i.e., a metric space $(\CS{M},d)$ and the multiset of demand points $\CS{U}$. Let $\CS{C}^*_1,\dots,\CS{C}^*_t$ be the clusters in an optimal solution, and let $n_j = |\CS{C}^*_j|$.
\item For each cluster $\CS{C}^*_j$, the adversary chooses a subset $\CS{A}^*_j \subseteq \CS{C}^*_j$ of cardinality $\lfloor (1-\rho)  \cdot n_j \rfloor$ to arrive in adversarial order, and let $\CS{R}^*_j  = \CS{C}^*_j \setminus \CS{A}^*_j$ be the subset of remaining demand points. We call the demand points in $\CS{A}^*_j$ adversarial-order points, and the points in $\CS{R}^*_j$ random-order points.\label{step:2}

\item The adversary orders the demand points in $\cup_{i=1}^{t} \CS{A}^*_i$. Then, for each cluster $\CS{C}^*_j$, the points in $\CS{R}^*_j$ are randomly positioned between the demand points in $\CS{A}^*_j$. That is, the relative position of $u \in \CS{R}^*_j$ among the points in $\CS{A}^*_j$ is chosen uniformly at random. More concretely, for each demand point $u \in \CS{R}^*_j$ an index position $s \in \{0,\dots, |\CS{A}^*_j|\}$ is drawn uniformly at random. Then, in the online sequence, $u$ must be positioned by the adversary between the $s$th and $(s+1)$th adversarial points in $\CS{A}^*_j$ (if $s = |\CS{A}^*_j|$, $u$ must be positioned after all the demand points in $\CS{A}^*_j$. The absolute positions of the demand points in $\cup_{i=1}^{t} \CS{R}^*_i$ in the online sequence are then chosen by the adversary, while keeping the relative order between the demand points in each cluster.
\end{enumerate}

The definition of the competitive-ratio in this setting is standard. An algorithm $\ALG$ is called $c$-competitive if for any input instance $\cI = (\CS{M},d,\demands, \CS{A}^*_1,\dots,\CS{A}^*_t)$, $\E{\ALG(\cI)} \leq c \cdot \OPT$, where the expectation is taken over the random positioning of the random-order demand points $\CS{R}^*_1,\dots,\CS{R}^*_t$, and the internal randomness of the algorithm.

To analyze $\rofl$ in the $\rho$-partial random-order setting, we bound the expected cost of $\rofl$ for the adversarial-order points in terms of the cost for the random-order points. As before, we focus our attention on a single cluster of $\OPT$, $\optcluster$. The idea is simple: Our analysis from Section~\ref{sec:improve} applies to the random-order demand points $\CS{R}^*$ in the cluster. Then, for each adversarial-order point $u \in \CS{A}^*$, if $u$ arrives after a random-order demand point $u' \in \CS{R}^*$, we can upper bound the distance from $u$ to its closest open facility by $\distalg{u'} + d^*_{u'} + d^*_{u}$. Otherwise, if there are no demand points in $\CS{R}^*$ that arrive before $u$, we can simply bound the cost paid for the service of $u$ by $1$.\footnote{We note that to get the upper-bound of $1$ on the service cost, we need to choose the piecewise function $g(\distletter) = q \cdot \distletter$ if $\distletter \leq 1$ and $g(\distletter) = 1$ otherwise, instead of $g(\distletter) = \min\{q\cdot \distletter,1\}$.} We formalize this intuition in the next theorem.

\begin{theorem}\label{thm:rho_parital}
for $\rho \in (0,1)$, $\rofl$ is $ (1+q)\max\{3/\rho - 1,(2/\rho -1)/q\}$-competitive in the $\rho$-partial random-order setting.
\end{theorem} 

\begin{proof}
Let $u_{1},\dots,u_{k}$ be the demand points in $\CS{R}^*$ ordered by their arrival order. We have by Theorem~\ref{thm:qrofl-competitive} that $\E{\ALG(\CS{R}^*)} \leq 1 + 1/q + 2 (1+q) \sum_{u\in \CS{R}^*} d^*_u$. For the adversarial-order demand points, we bound $\E{\sum_{u \in \CS{A}^*} \coin{u}}$ (recall that $\coin{u} = \min\{ q\cdot \distalg{u},1\}$). We first consider the distance from the optimal center $\optcenter$ to the closest open facility at the online rounds in which adversarial-order demand points arrive. Let $J_{\CS{A}} \subseteq [n]$ be the set of online rounds in which an adversarial demand point from $\CS{A}^*$ arrives, and let $Z \in J_{\CS{A}}$ be a random variable that gets a uniformly random online round in $J_{\CS{A}}$. Randomly positioning the random-order online rounds within the adversarial-order online rounds is equivalent to randomly positioning the adversarial-order online rounds within the random-order online-rounds. Hence, the probability that $Z$ arrives between $u_j$ and $u_{j+1}$ is $1/(k + 1)$, and in this case we have $d(F_{Z-1}, c^*) \leq \distalg{u_j} + d^*_{u_j}$, and so, $\min\{q \cdot d(F_{Z-1}, c^*),1\} \leq \min\{q\cdot \distalg{u_j},1\} + q\cdot d^*_{u_j} = \coin{u_j} + q\cdot d^*_{u_j}$. We get that 
\begin{align*}
    \E{\min\{q\cdot d(F_{Z-1}, c^*),1\}} \leq \frac{1}{k+1} + \frac{1}{k+1} \sum_{j=1}^{k} \left( \E{ \coin{u_j}} + q\cdot d^*_{u_j}\right).
\end{align*}
Now, we can simply upper bound $\distalg{v_\ell} \leq d(F_{\ell-1},c^*) + d^*_{v_\ell}$, and so,
\begin{align}
    \E{\sum_{u \in \CS{A^*}}\coin{u}} &= \E{\sum_{u \in \CS{A^*}}\min\{q\cdot \distalg{u},1 \}} \notag \\
    & \leq \E{\sum_{\ell \in J_{A}} \min\{q \cdot (d(F_{\ell-1}, c^*) + d^*_{v_\ell}) ,1 \}}
    \notag \\
    &\leq  \E{\sum_{\ell \in J_{A}} \min\{q \cdot d(F_{\ell-1}, c^*) ,1 \}} + q\sum_{u\in \CS{A}^*} d^*_{u}
    \notag \\
    &\leq 
    |\CS{A^*}| \cdot \left(\frac{1}{k+1} + \frac{1}{k+1} \sum_{j=1}^{k} \left( \E{ \coin{u_j}} + q\cdot d^*_{u_j}\right) \right) + q \sum_{u \in \CS{A}^*} d^*_u \label{eq:line_robust}\\
    &\leq \left(1/\rho - 1\right) \left(1+  \E{\sum_{u \in \CS{R}^*} \coin{u} } + q\sum_{u \in \optcluster} d^*_u \right), \notag
\end{align}
where in the last inequality, we used the fact that $|\CS{A}^*|/(k+1) = |\CS{A}^*|/(|\CS{R}^*|+1) \leq 1/\rho - 1$.
Now, similarly to Theorem~\ref{thm:qrofl-competitive} we have $ \E{\sum_{u \in \CS{R}^*} \coin{u} } \leq 1 + 2q \sum_{u \in \CS{R}^*}d^*_u$, together with Lemma~\ref{lem:price_factor_roflq}, we get that $\E{\ALG(\CS{A}^*)} \leq (1+1/q)(1/\rho -1)(2 + 3q \sum_{u \in \optcluster} d^*_u)$. Finally, we add the cost paid for the service of the points in $\CS{R}^*$, and get that
\begin{align*}
    \E{\ALG(\optcluster)} &= \E{\ALG(\CS{A}^*)} + \E{\ALG(\CS{R}^*)} \\ 
    &\leq 
    \left(1+\frac{1}{q}\right)\left(\frac{1}{\rho} -1\right)\left(2 + 3q \sum_{u \in \optcluster} d^*_u\right) + 1 + \frac{1}{q} + 2 (1+q) \sum_{u\in \CS{R}^*} d^*_u \\
    &\leq \left(1+\frac{1}{q}\right)\left(\frac{2}{\rho} -1\right) + (1+q)\left( \frac{3}{\rho} - 1\right) \sum_{u\in \optcluster} d^*_u. \qedhere
\end{align*}
\end{proof}

We note that in a similar setting, where the adversarial-order demand points are not chosen by the adversary, but rather drawn randomly, we can obtain a better upper bound on the performance of $\rofl$. 

More formally, consider Step~\ref{step:2} in the random process that generates the online sequence in the $\rho$-partial random-order setting, and consider the case where for all $j\in [t]$, $A^*_j \subseteq \CS{C}^*_j$ is a uniformly random subset of cardinality $\lfloor (1-\rho) n_j \rfloor$ (instead of an adversarially chosen subset of the same cardinality), and $R^*_j = \CS{C}^* \setminus A^*_j$. We refer to this setting by \textit{$\rho$-partial random-order with random adversarial-order points}. With small adaptations in the proof of Theorem~\ref{thm:rho_parital}, and by using the fact that in this setting $\E{\sum_{u \in R^*} d^*_u} = \frac{|R^*|}{|\optcluster|} \sum_{u \in \optcluster} d^*_u $ and $\E{\sum_{u \in A^*} d^*_u} = \frac{|A^*|}{|\optcluster|} \sum_{u \in \optcluster} d^*_u $ we get the following result.

\begin{theorem}\label{thm:rho_partial_random}
for $\rho \in (0,1)$, $\rofl$ is $ (1+q)\max\{4-2\rho,(2/\rho -1)/q\}$-competitive in the $\rho$-partial random-order setting with random adversarial-order points.
\end{theorem} 

The proof of Theorem~\ref{thm:rho_partial_random} is given in Appendix~\ref{proof:rho_partial_random}. For example, Theorem~\ref{thm:rho_partial_random} shows that $\rofl[]$ is $6$-competitive when a random half of the demand points in each cluster arrive in adversarial order. For comparison, if half of the demand points in each cluster are chosen adversarially, the upper bound that we get from Theorem~\ref{thm:rho_parital} on the competitive-ratio of $\rofl[]$ is only $10$.
\section{Discussion}\label{sec:discussion}

In this work, we resolve the open question regarding the true performance of Meyerson's algorithm ($\rofl[]$) in the random-order model. Furthermore, we introduce a general family of algorithms in the form of $\rofl[]$, and derive the best algorithm in this family, $\rofl[\nicefrac{1}{2}]$, which achieves the state-of-the-art performance. We prove that $\rofl[\nicefrac{1}{2}]$ is $3$-competitive and that the best possible competitive-ratio for this problem is $2$. 

Several interesting questions remain open for future research.
First, the true performance of Meyerson's algorithm for non-uniform facility costs in the random-order model remains open. It would be interesting to see if our techniques can be used to obtain tight analysis for this case too. Another interesting direction is to study the performance of the simple deterministic algorithm by Fotakis~\cite{fotakis2007primal} in the random-order model, which is still unknown. We note that a slight modification of the instance in the proof of our lower bound (Theorem~\ref{thm:lower_bound}) shows that the competitive-ratio of Fotakis' algorithm is no better than $3$. We prove this result in Appendix~\ref{apx:fotakis}.

Finally, a gap between the lower and upper bounds that we obtain for the facility location problem with uniform facility costs remains open. To the best of our knowledge, there are no candidate online algorithms in the literature which could outperform $\rofl[\nicefrac{1}{2}]$ and beat the competitive-ratio of $3$. Hence, if the competitive-ratio of $3$ is not optimal, new algorithmic ideas are needed to beat this bound. Additionally, since our lower bound holds in the weaker $\IID$ model (with full prior knowledge of the distribution), it would be interesting to study whether a competitive-ratio of $2$ can be achieved in the $\IID$ model.

\bibliographystyle{plain}
\bibliography{citations}

\appendix
\section{Proof of Theorem~\ref{thm:rho_partial_random}}\label{proof:rho_partial_random}
\begin{proofw}
The proof is very similar to the proof of Theorem~\ref{thm:rho_parital}. Let $n' = |\CS{C}^*|$, $k = \lceil \rho n' \rceil$ and $a = \lfloor (1-\rho) n' \rfloor$. As opposed to the proof of Theorem~\ref{thm:rho_parital}, in this setting, $R^*$ is a random subset. Hence, we condition on $R^* = \CS{Y}$ for some fixed subset $\CS{Y} \subseteq \CS{C}^*$ of cardinality $k$. The proof proceeds as the proof of Theorem~\ref{thm:rho_parital}, until we reach Inequality~\eqref{eq:line_robust}. Conditioned on the event $\{R^* = \CS{Y}\}$, we have $A^* = \CS{C}^* \setminus \CS{Y}$ and by substituting $|A^*| = a$ in Inequality~\eqref{eq:line_robust}, we get that
\begin{align*}
    \E{\given{\sum_{u \in \CS{A^*}}\coin{u}}{R^* = \CS{Y}}} &\leq
    a \cdot \left(\frac{1}{k+1} + \frac{1}{k+1} \sum_{j=1}^{k} \left( \E{ \coin{u_j}} + q\cdot d^*_{u_j}\right) \right) + q \sum_{u \in \CS{C}^* \setminus \CS{Y}} d^*_u \\
    &= \frac{a}{k+1} \left(1+  \E{\sum_{u \in \CS{Y}} \coin{u} } + q\sum_{u \in \CS{Y}} d^*_u \right) + q \sum_{u \in \CS{C}^* \setminus \CS{Y}} d^*_u , 
\end{align*}

Now, similarly to Theorem~\ref{thm:qrofl-competitive} we have $ \E{\given{\sum_{u \in \CS{Y}} \coin{u}}{R^* = \CS{Y}} } \leq 1 + 2q \sum_{u \in \CS{Y}}d^*_u$, together with Lemma~\ref{lem:price_factor_roflq}, we get that $\E{\given{\ALG(A^*)}{R^* = \CS{Y}}} \leq (1+1/q)\frac{a}{k+1}(2 + 3q \sum_{u \in \CS{Y}} d^*_u) + (1+q) \sum_{u \in \CS{C}^* \setminus \CS{Y}} d^*_u$. Finally, we add the cost paid for the service of the points in $R^*$, and use the bound $\E{\given{\ALG(R^*)}{R^* = \CS{Y}}} \leq 1 + 1/q + 2 (1+q) \sum_{u\in \CS{Y}} d^*_u$. We have 
\begin{align*}
    & \E{\given{\ALG(\optcluster)}{R^* = \CS{Y}}} \\
    & \quad = \E{\given{\ALG(A^*)}{R^* = \CS{Y}}} + \E{\given{\ALG(R^*)}{R^* = \CS{Y}}} \\ 
    & \quad \leq  \left(1+\frac{1}{q}\right)\frac{a}{k+1}\left(2 + 3q \sum_{u \in \CS{Y}} d^*_u  \right) + (1+q) \sum_{u \in \optcluster \setminus \CS{Y}} d^*_u +  1 + \frac{1}{q} + 2 (1+q) \sum_{u\in \CS{Y}} d^*_u\\
    & \quad = 
    \left(1+\frac{1}{q}\right)\left(\frac{2a}{k+1} + 1\right) + (1+q)\left( \frac{3a}{k+1} + 2 \right) \sum_{u\in \CS{Y}} d^*_u + (1+q) \sum_{u \in \optcluster \setminus \CS{Y}} d^*_u.
\end{align*}
By taking the expectation over $R^*$, we get that
\begin{align}
\begin{split}
    \E{\ALG(\optcluster)}  &\leq 
    \left(1+\frac{1}{q}\right)\left(\frac{2a}{k+1} + 1\right) + (1+q)\left( \frac{3a}{k+1} + 2 \right) \E{\sum_{u\in R^*} d^*_u} \\
    & \quad  + (1+q) \E{ \sum_{u \in A^*} d^*_u}.
\end{split}\label{eq:apx_r*}
\end{align}
Now since $R^* \subseteq \optcluster$ is a uniformly random subset of cardinality $k$, we have $\E{\sum_{u\in R^*} d^*_u} = \frac{k}{n'}\sum_{u\in \optcluster} d^*_u $, and likewise,  $\E{\sum_{u\in A^*} d^*_u} = \frac{a}{n'}\sum_{u\in \optcluster} d^*_u $. By substituting these two equations in Equation~\eqref{eq:apx_r*}, we obtain
\begin{align*}
    \E{\ALG(\optcluster)} & \leq \left(1+\frac{1}{q}\right)\left(\frac{2a}{k+1} + 1\right) + (1+q)\left( \frac{3a}{k+1} \frac{k}{n'}+ \frac{2k}{n'}+ \frac{a}{n'}
    \right) \sum_{u\in \optcluster} d^*_u  \\
    &\leq \left(1+\frac{1}{q}\right)\left(\frac{2a}{k+1} + 1\right) + (1+q)\left( \frac{4a + 2k}{n'} \right) \sum_{u\in \optcluster} d^*_u \\
    & \leq  \left(1+\frac{1}{q}\right)\left(\frac{2}{\rho} -1\right) + (1+q) \left(  4 - 2\rho  \right) \sum_{u\in \optcluster} d^*_u,
\end{align*}
where in the last inequality we used the fact that $2a/(k+1) \leq 2(1/\rho - 1)$ and $(4a + 2k)/n' = (4n' - 2k)/n' \leq 4-2\rho$. 
\end{proofw}

\section{A Lower Bound on the Random-Order Competitive-Ratio of Fotakis' Algorithm}\label{apx:fotakis}

Fotakis' algorithm maintains a potential for each point $z$ in the metric space $\CS{M}$. The potential of $z$ at online round $\ell$ is defined by $p_\ell(z) = \sum_{i = 1}^{\ell} \max\{d(F_{\ell-1}, v_i) - d(z,v_i), 0\}$. The algorithm operates as follows: When a demand point $v_\ell$ arrives at online round $\ell$, the algorithm computes the potentials $p_\ell(z)$ of all points $z \in \CS{M}$. Then, it considers the point $z_\ell$ with the largest potential (ties are broken arbitrarily). If $p_\ell(z_\ell) \geq 1$, it opens $z_\ell$. Then, it assigns the demand point $v_\ell$ to its closest open facility (see~\cite{fotakis2007primal} and \cite{nagarajan2013offline} for more details).\footnote{For convenience, we choose to describe the facility opening criterion with a weak inequality as in~\cite{nagarajan2013offline} (i.e., $p_\ell(z_\ell) \geq 1$) instead of a strict inequality as in~\cite{fotakis2007primal}.}

To obtain the lower bound on the competitive-ratio of Fotakis' algorithm, we modify the constructed instance in the proof of Theorem~\ref{thm:lower_bound} as follows. For $0 < \delta < 1$, we multiply all the distances in the metric space by a factor of $\delta$. For our purpose, it also suffices to choose $m = 2n$ (instead of $m=n^2$). Concretely, for $m = 2n$ we construct a metric space with $m + \binom{m}{n}$ points of two types: The first type consists of $m$ points, $x_1,\dots,x_m$, with $d(x_i,x_j) = \delta$ for all $i,j\in [m]$. The second type are the \textit{subset points}. For each subset $\CS{I} \subseteq [m]$ of cardinality $n$, there is a point $s_{\CS{I}}$ with $d(s_\CS{I},x_j) = \delta/2$ if $j \in \CS{I}$ and $d(s_\CS{I},x_j) = \delta$ otherwise. Finally, for two subset points $s_\CS{I} \neq s_\CS{J}$, $d(s_\CS{I},s_\CS{J}) = \delta$. 
For the input, we can simply take $\CS{U} = \{x_1,\dots,x_n\}$.

To upper bound the cost of $\OPT$, observe that $\CS{U}$ is a subset of the points in $\{x_1,\dots,x_m\}$ of cardinality $n$, and therefore, there is a subset point $s_{\CS{I}}$ at distance $\delta/2$ from all the demand points in $\CS{U}$. Hence, $\E{\OPT} \leq 1 + n (\delta/2)$.

We now consider the performance of Fotakis' algorithm which we denote by $\ALG$. When the first demand point $v_1$ arrives, $\ALG$ opens a facility at $v_1$. Then, the algorithm does not open additional facilities until a point has a potential of at least $1$. First, note that except for the first demand point $v_1$, the points in $\{x_1,\dots,x_m\}$ never have a potential $ \geq 1$. This is because after $v_1$ is opened, each demand point $v_\ell$ (for $\ell > 1$) is at distance at most $\delta$ from its closest open facility, i.e., $d(F_{\ell-1}, v_\ell) \leq \delta$. So it does not contribute to the potential of $\{x_1,\dots,x_m\}$ except for its own potential, for which it contributes $\delta < 1$. Hence, except for $v_1$, no facilities in $\{x_1,\dots,x_m\}$ are opened by the algorithm.

On the other hand, the subset points may accumulate a potential of $1$. When a demand point $v_\ell = x_j$ arrives, and $d(F_{\ell-1},v_\ell) = \delta$, it contributes $\delta/2$ to the potential of each subset point $s_{\CS{I}}$ such that $j \in \CS{I}$. Therefore, at online round $\ell = 2/\delta + 1$, there are subset points $s_{\CS{I}}$ such that $d(v_{i},s_{\CS{I}}) = \delta/2$ for all $i \leq \ell$, and so $p_\ell(s_{\CS{I}}) = (2/\delta) \cdot (\delta/2) = 1$.\footnote{For convenience, we assume that $2/\delta$ and $(n-1) \delta/2$ are integers.} Since the algorithm chooses to open a facility at an arbitrary subset point $s_{\CS{I}}$ with $p_\ell(s_{\CS{I}}) = 1$, we can assume that it chooses $s_{\CS{I}}$ which is close only to the demand points that arrived until round $\ell$, and at distance $\delta$ from all future demand points $v_{\ell+1},\dots,v_{n}$ (for instance, we can choose $\CS{I} = \{j : v_i = x_j, i \leq \ell\} \cup \{m,m-1,\dots,m - (n- \ell) + 1\}$). Then, the potential of all the subset points return to zero, and this process is repeated every $2/\delta$ online rounds.

To sum up, the algorithm opens $1 + (n-1)\delta/2$ facilities. For the assignment cost it pays $0$ for $v_1$, $\delta/2$ for all the demand points that arrive at online rounds in which a facility is opened, and $\delta$ for all other demand points. Hence, the total assignment cost is $(n-1) \frac{\delta}{2} \cdot \frac{\delta}{2} + (n - 1 - (n-1) \frac{\delta}{2}) \cdot \delta $. Overall, the algorithm pays $1 + (n-1) \frac{3 \delta}{2} - (n-1)\frac{\delta^2}{4} $. By taking $\delta = 1/\sqrt{n-1}$, we get that $\E{\ALG}/\E{\OPT}$ approaches $3$ as $n$ approaches infinity.

\end{document}